\documentclass[a4paper,UKenglish]{lipics-v2021}

\usepackage{microtype}
\usepackage{algorithm,algorithmicx,algpseudocode}
\usepackage{cleveref}

\long\def\longdelete#1{}

\newcommand{\Xpred}{\hat{\mathbf{x}}}
\newcommand{\hatT}{\hat{T}}
\newcommand{\hatt}{\hat{t}}
\newcommand{\X}{\mathbf{x}}
\newcommand{\xhat}{\hat{x}}
\newcommand{\defeq}{\mathrel{\mathop:}=}

\newcommand{\costalg}{Z^{ALG}}
\newcommand{\costopt}{Z^{OPT}}
\newcommand{\errortime}{\varepsilon_{time}}
\newcommand{\errorpos}{\varepsilon_{pos}}
\newcommand{\errorlast}{\varepsilon_{last}}

\newcommand{\PAH}{\textsc{PAH}}
\newcommand{\algNI}{\textsc{LAR-NID}}
\newcommand{\algI}{\textsc{LAR-ID}}
\newcommand{\algT}{\textsc{LAR-Trust}}
\newcommand{\algL}{\textsc{LAR-Last}}
\newcommand{\redesign}{\textsc{Redesign}}
\newcommand{\replan}{\textsc{Replan}}
\newcommand{\algdarNI}{\textsc{LADAR-NID}}
\newcommand{\algdarI}{\textsc{LADAR-ID}}
\newcommand{\algdarT}{\textsc{LADAR-Trust}}
\newcommand{\algdarL}{\textsc{LADAR-Last}}
\newcommand{\algsmartstart}{\textsc{SmartStart}}

\newenvironment{psketch}{%
\renewcommand{\proofname}{Proof Sketch}\proof}{\endproof}

\algrenewcommand\algorithmicrequire{\textbf{Input:}}
\algrenewcommand\algorithmicensure{\textbf{Output:}}

\usepackage{thm-restate}

\title{Online TSP with Predictions}

\author{Hsiao-Yu Hu}{Institute of Information Science, Academia Sinica, Taiwan\\Department of Industrial Engineering and Engineering Management,
National Tsing Hua University, Hsinchu 30013, Taiwan}{hsiaoyu.hu@gapp.nthu.edu.tw}{}{}

\author{Hao-Ting Wei}{Department of Industrial Engineering and Operations Research, Columbia University, United States}{hw2738@columbia.edu}{}{}

\author{Meng-Hsi Li}{Department of Industrial Engineering and Engineering Management,
National Tsing Hua University, Hsinchu 30013, Taiwan}{ted19980924@gapp.nthu.edu.tw}{}{}

\author{Kai-Min Chung}{Institute of Information Science, Academia Sinica, Taiwan}{kmchung@iis.sinica.edu.tw}{}{}

\author{Chung-Shou Liao}{Department of Industrial Engineering and Engineering Management,
National Tsing Hua University, Hsinchu 30013, Taiwan}{csliao@ie.nthu.edu.tw}{}{}

\authorrunning{H.\,Y. Hu, H.\,T. Wei, M.\,H. Li, K.\,M. Chung, and C.\,S. Liao} 

\Copyright{H.\,Y. Hu, H.\,T. Wei, M.\,H. Li, K.\,M. Chung, and C.\,S. Liao} 



\category{} 

\relatedversion{} 

\supplement{}

\acknowledgements{}

\EventEditors{}
\EventNoEds{0}
\EventLongTitle{}
\EventShortTitle{}
\EventAcronym{}
\EventYear{}
\EventDate{}
\EventLocation{}
\EventLogo{}
\SeriesVolume{00}
\ArticleNo{00}

\begin{document}

\nolinenumbers

\maketitle


\longdelete{
With the increasing reliance on online delivery platforms, efficiently designing a route to deal with the requests in an online fashion becomes a primary issue. In this study,  
}

\begin{abstract}
We initiate the study of online routing problems with predictions, inspired by recent exciting results in the area of learning-augmented algorithms.
A learning-augmented online algorithm
which incorporates predictions in a black-box manner to outperform existing algorithms if the predictions are accurate while otherwise maintaining theoretical guarantees 
is a popular framework for overcoming pessimistic worst-case competitive analysis. 

In this study, 
we particularly begin investigating
the classical online traveling salesman problem (OLTSP), where future requests are augmented with predictions. 
Unlike the prediction models in other previous studies, each actual request in the OLTSP, associated with its arrival time and position, may not coincide with the predicted ones, which, as imagined, leads to
a troublesome situation. 
Our main result is to study different prediction models 
and design algorithms to improve the best-known results in the different settings. 
Moreover, we generalize the proposed 
results to the online dial-a-ride problem.

\end{abstract}

\newpage

\section{Introduction}
In many applications, people make decisions without knowing the future, and two approaches are widely used to address this issue: machine learning and online algorithms.
While machine learning methods usually lack theoretical guarantees, online algorithms are often evaluated based on the ratio of the optimal offline cost to the cost achieved by online algorithms over 
worst-case
instances (i.e., the competitive ratio).

In recent years, a rapidly growing field of research known as learning augmentation has attempted to 
merge the above two approaches: 
incorporating machine-learned predictions into online algorithms 
from a theoretical point of view. 
In this line of work, such an algorithm is given some type of prediction 
to the input, but the prediction is not always accurate. The goal of these studies is to design novel algorithms that have the following three properties: (1) achieve good performance when the 
given 
prediction is perfectly accurate, which is called {\it consistency}, (2) maintain worst-case bounds when the prediction is terribly wrong, which is called {\it robustness}, and (3) the performance should not degrade significantly when the prediction is slightly inaccurate, which is called {\it smoothness}. As the competitive ratio of a learning-augmented algorithm depends on errors, we define it as a function 
$c$
of error $\varepsilon$ using 
a given predictor 
such that $c(\varepsilon)$ represents the worst-case competitive ratio given 
that the error is at most $\varepsilon$. Note that the definition of $\varepsilon$ is subject to the problem setting (as shown in Section 2).
Optimization problems that have been studied under this framework include online paging \cite{jiang2020online, lykouris2021competitive, rohatgi2020nearoptimal, wei2020better}, ski-rental \cite{anand2020customizing, gollapudi2019online, kumar2018improving, wei2020optimal}, scheduling \cite{bamas2020learning, im2021nonclairvoyant, lattanzi2020online, kumar2018improving, wei2020optimal}, matching \cite{AGKK20}, bin packing \cite{angelopoulos2019online}, queueing \cite{mitzenmacher2021queues}, secretary \cite{AGKK20, dutting2021secretaries}, covering \cite{bamas2020primal}, bidding \cite{medina2017revenue}, knapsack \cite{im2021online}, facility location \cite{FGGP21}, and graph \cite{azar2022online} problems.

However, 
in addition to 
the above online optimization problems, no existing studies extended the learning-augmented framework to 
online routing problems, accompanied by many real-world applications. 
For example, due to COVID-19, 
people have become more reliant on online food 
delivery 
platforms such as Seamless 
and 
Uber Eats. As a result, optimizing the delivery route 
is 
an important issue for the platform to win over such a huge market. 
Many platforms have used machine learning techniques to improve their 
services, 
but most studies did not provide worst-case guarantees. In this study, we look at how the learning-augmented framework can be applied to 
online routing problems.
Note that 
previous research in online learning also studied the combination of online algorithms and machine learning for routing problems; 
readers may refer to the survey~\cite{blum1998online} for more details. 
However, our study aims at designing a route with predictions which learned from historical 
information rather than 
choosing the best possible 
strategy 
adaptively 
based on some given routes. 
In particular, 
we study the online TSP problem proposed by Ausiello et.al.~\cite{ausiello2001algorithms} 
and then extend to the 
online dial-a-ride problem.

Let the offline optimal cost and the cost achieved by an online algorithm 
be denoted by $\costopt$ and $\costalg$, respectively. 
To measure 
the performance of an learning-augmented online algorithm, we follow the definitions proposed in \cite{lykouris2021competitive,kumar2018improving}. 

\begin{definition}[Consistency]\label{def:con}
An algorithm is said to be $\alpha$-consistent if $c(0)=\alpha$.
\end{definition}
\begin{definition}[Robustness]\label{def:rob}
An algorithm is said to be $\beta$-robust if $c(\varepsilon) \leq \beta$ for 
any 
$\varepsilon$.
\end{definition}
\begin{definition}[Smoothness]\label{def:smo}
Given an $\alpha$-consistent algorithm, 
it is said to be $f(\varepsilon)$-smooth if 
$\costalg \leq \alpha \cdot \costopt + f(\varepsilon)$ holds on any input for some 
continuous
function $f$ such that $f(0) = 0$.
\end{definition}

\noindent{\textbf{The OLTSP.}}
In this study, we first consider one of the classical online combinatorial optimization problems, the online traveling salesman problem (OLTSP), in which the input is a sequence of requests that arrive 
in an online fashion. 
A salesman, the server, is out to visit every request after it arrives and eventually return to the origin 
such 
that the completion time is minimized. 
Ausiello et.al.~\cite{ausiello2001algorithms} showed that the problem has a lower bound of 2 and presented 
an optimal 2-competitive algorithm, called \textsc{Plan-At-Home (PAH)}. 
Meanwhile, 
the currently best polynomial-time algorithm is approximately $2.65$-competitive~\cite{ascheuer2000online}. 

\medskip
\noindent{\textbf{The OLDARP.}}
We also consider the (uncapacitated) online dial-a-ride problem (OLDARP), which is a generalization of the OLTSP and has the same lower bound of 2. 
The key
difference between the two problems is that a server needs to transport each request from its 
pickup position
to its
delivery position. 
An optimal 
$2$-competitive algorithm 
was
proposed in~\cite{ascheuer2000online}.

\medskip
\noindent{\textbf{The OLTSP with Predictions.}}
The first 
question for
developing learning-augmented algorithms is: 
``what kind of predictions 
are required
for the problem ?'' 
We refer to the discussion in \cite{lattanzi2020online} to answer this question: a good prediction should be learnable; that is, it should be possible to learn the prediction from historical data with an appropriate sample complexity. 
For the OLTSP, 
it is natural to forecast its future requests.
To show the learnability, one can set up a learning task and   
exploit some 
different features 
of historical data, e.g., position of the requests, 
arrival time of the requests, and the number of total requests,
to learn a function; for instance, to minimize the mean square error (mse).
In this study, we 
thus 
investigate three models with different types of predictions. 
Notice that 
the significant difference between these predictions
is whether the 
input
scale 
is known 
for learning.
In the first two models, 
intuitively, we predict the whole sequence. 
Formally, each request in the prediction has its 
own predicted arrival time and position. 
We split into two cases: \emph{the sequence prediction without identity} is a predicted sequence with an arbitrary size 
over requests, and \emph{the sequence prediction with identity} is a predicted sequence that contains the exact same number of requests as the actual input. These two sequence predictions represent different learning models; the former one views the sequence of requests as a group, while the latter one examines each of the requests individually.
Last, we consider 
a special prediction model, called 
\emph{the prediction of the last arrival time}, where the server receives the least amount of information 
in 
the three models. The reason of choosing 
such a 
prediction is that the arrival time of the last request actually provides a 
lower bound for the optimal route and that this information is enough to help the server in a polynomial setting. 
In the following sections, we show how to design learning-augmented online algorithms by 
appropriately incorporating 
the predictions into
online algorithms. 
As expected, predicting a sequence of requests provides more information about future events, 
certainly 
leading to a relatively higher cost than predicting a single request. 
That is, 
the choice of predictions can be read as a trade-off between the performance of online algorithms and the sample complexity.

\medskip
\noindent{\textbf{Our Contributions.}}
The key
contribution is to develop three models for the OLTSP in each of which we present a learning-augmented algorithm. 
The results are as follows:
\begin{itemize}
    \item [1.] Consider the sequence prediction without identity. For this arbitrary sequence case, we design an algorithm that is 
    $(1+\lambda)$-consistent and $(3+2/ \lambda)$-robust, where $\lambda \in (0,1]$ is a parameter describing the confidence level in the prediction.
    \item [2.] Consider a restricted model in which the sequence prediction has the same size as the set of actual requests. We propose a $\min \{3, 1 + (2\errortime + 4 \errorpos) / \costopt\}$-competitive algorithm, where $\errortime = max_{i \in [n]}|\hat{t}_i - t_i|$ is the maximum time difference between the predicted and actual requests and $\errorpos = \sum_{i}d(\hat{p}_i, p_i)$ is the accumulated distance between the predicted and actual positions. 
    \item [3.] Consider the prediction of the last arrival time. We design a $\min \{4, 2.5 + |\errorlast| / \costopt \}$-competitive polynomial-time algorithm, where $\errorlast = \hatt_n - t_n$ 
    denotes 
    the difference between the predicted and actual arrival time of the last request.
\end{itemize}
\smallskip
We also study the lower bounds for the OLTSP with predictions:
\smallskip
\begin{itemize}
    \item [1.] Consider the sequence prediction without identity. Any $1$-consistent algorithm has no robustness. We show there is a trade-off between consistency and robustness.
    \item [2.] Consider the sequence prediction with identity. Any $1$-consistent algorithm has a robustness that is at least 2.
    \item [3.] Consider the prediction of the last arrival time. We show that the consistency cannot be better than 2; 
    that is, 
    partial information is not enough for breaking the original lower bound of the OLTSP.
\end{itemize}
In addition, we extend our models to the 
online
dial-a-ride problem.

\section{Preliminaries}\label{premlim}
In this section, we first give 
the 
formal definition of the OLTSP, and 
revisit the \PAH~algorithm, a key ingredient for 
designing our algorithms. Finally, we define error measurement for the 
three
models.

\subsection{The OLTSP}
Recall the metric OLTSP in which the input is a sequence of $n$ requests $\X = (x_1, \ldots, x_n)$
in a metric space $M$. 
Each request $x_i \in \X$ can be represented 
by $x_i = (t_i, p_i)$, where $t_i$ denotes the arrival time
and $p_i$ denotes the position of the request. In addition, we assume the requests are in 
non-decreasing order of arrival time; that is, $t_i \leq t_j$ for any $i < j$, and $t_n$ is the arrival time of the last request in $\X$. A server is required to start at the origin $o$ at time $0$, visit each request $x_i$ in $\X$ at position $p_i$ no earlier than its arrival time $t_i$, and at last come back to the origin $o$, which is also called home. Assuming the server moves with unit speed, the goal is to find a route $T_{\X}$ such that the completion time, denoted by $|T_{\X}|$, is minimized.

\longdelete{
Regarding 
the previous results of the metric OLTSP without predictions, Ausiello et al.~\cite{ausiello2001algorithms} first studied this problem; they showed the lower bound of 2 and proposed an optimal algorithm, \PAH. Note that 
the \PAH~algorithm can yield a competitive ratio of 3 in polynomial time 
using Christofides’ heuristic~\cite{worboys_1985}. 
Ausiello et al.~\cite{ausiello2004algorithms} improved the polynomial-time result to $\frac{7+\sqrt{17}}{4}\approx 2.78$;  
on the other hand, 
Ascheuer et al.~\cite{ascheuer2000online} obtained the currently best 
ratio of 
$\frac{7+\sqrt{13}}{4}\approx 2.65$ by leaving the server idle from time to time smartly.}

\medskip 
\noindent{\textbf{The \PAH~Algorithm~\cite{ausiello2001algorithms}.}}
Before introducing the proposed online algorithms with predictions, we first revisit 
the 
online algorithm, \textsc{Plan-At-Home (PAH)} for the OLTSP in a metric space, reported in~\cite{ausiello2001algorithms}, and see how we can attain 
the robustness of our algorithms. The greedy \PAH~algorithm achieves a competitive ratio of $2$ when assuming the server can access an optimal solution 
to visit a set of released requests. Although the optimal route cannot be computed in polynomial time unless $P=NP$, 
the \PAH~algorithm is $3$-competitive by using Christofides' heuristic~\cite{worboys_1985} to approximate the optimal route. 
The details of the \PAH~algorithm can be referred to the Appendix~\ref{PAH:code}.

\begin{theorem}[\cite{ausiello2001algorithms}, Theorem 4.2]\label{thm:PAH2}
The \PAH~algorithm is $2$-competitive for the OLTSP in a metric space.
\end{theorem}

\begin{theorem}[\cite{ausiello2001algorithms}, Theorem 5.3]\label{thm:PAH3}
The \PAH~algorithm is a $3$-competitive polynomial-time algorithm for the OLTSP in a metric space when using Christofides' heuristic.
\end{theorem}

The overview of the \PAH~algorithm has the following two operations: (1) find a route only when the server is at the origin $o$, and (2) act differently in response to 
the requests that are relatively close to the origin and those that are relatively far.
At moment $t$ during the execution of the algorithm, 
let $p(t)$ denote the position of the server 
and 
$U_t$ denote the set comprising 
every 
unserved request that has been present 
so far. 
Precisely, 
if the server is at the origin, i.e., $p(t)=o$, it starts to follow an optimal (or approximate) route $T_{U_t}$, which visits all released unserved requests in $U_t$
and returns to the origin $o$. 
Otherwise,
when the server is on a route
$T_{U_{t'}}$, where $t' < t$ and $t'$ denotes the 
moment when 
the route was designed, 
and a new request $x_i = (t_i, p_i)$ arrives, the server moves depending on the relationship between $d(p_i, o)$ and $d(p(t), o)$, where 
$d(a, b)$ denotes the shortest distance between points $a$ and $b$. If $d(p_i, o) > d(p(t), o)$, 
request $x_i$ is considered far from the origin, and the server terminates the route $T_{U_{t'}}$ and goes back to the origin $o$ directly. Otherwise, 
request $x_i$ is considered close to $o$, and 
thus the server ignores $x_i$ until it is back 
home 
and continues with 
the 
current
route $T_{U_{t'}}$. 

In addition, we can derive Corollary~\ref{cor:PAHwait} from the proof of \PAH~in~\cite{ausiello2001algorithms}.
\begin{corollary}\label{cor:PAHwait}
Assume an algorithm asks the server to wait at the origin $o$ for time $t$ before following the \PAH~algorithm where $t \leq t_n$. The algorithm is $2$-competitive if the server has access to the offline optimal route to serve a set of released requests. Also, it is $3$-competitive using Christofides' algorithm.
\end{corollary}

\subsection{Prediction Models and Error Measure}
In this section, we present three different models of prediction.
First, we consider 
predicting a sequence of requests $\Xpred = (\xhat_1, \ldots, \xhat_m)$, where $m$ denotes the number of predicted requests in the sequence. As the server is unaware of the size of $\X$ in the beginning, we assume that the learning model forecasts the number of requests before predicting the arrival time and positions of all requests. Formally, each request in the prediction is defined by $\xhat_i = (\hat t_i, \hat p_i)$, where $\hat t_i$ and $\hat p_i$ are its predicted arrival time and position. As we need to compare two request sequences with different sizes, $n$ and $m$, 
apparently 
a proper definition of prediction errors does not exist. Consequently, we only care about whether any error exists. That is, we simply distinguish whether $\Xpred = \X$ or not. In this model, we 
can 
design an 
algorithm that is consist and robust but not smooth. 

Next, we consider a restricted version 
in which 
the number of requests $n$ is given, 
implying that the prediction has the same size as the actual request sequence, i.e., $|\Xpred|=|\X|=n$. We 
thus 
assume that there is a one-to-one 
correspondence 
between the requests in the two sequences $\Xpred$ and $\X$. That is, each predicted request $\xhat_i$ is 
exactly 
paired 
with 
the actual request $x_i$. Since the prediction gives us both the arrival time and position, we define errors with respect to the two parameters. We first define the time error, denoted by $\errortime$, as the maximum difference between the arrival time of a predicted request and its corresponding actual request.
$$\errortime \defeq max_{i \in [n]}|\hat{t}_i - t_i|$$
Then, we define the position error, denoted by $\errorpos$, as the sum of distances between the positions of the predicted and actual requests.
$$\errorpos \defeq \sum^{n}_{i=1}d(\hat{p}_i, p_i)$$
In this restricted model, 
we extend the above algorithm to 
a new one which can 
improve the result and satisfy the smoothness requirement.

Last, we consider the prediction 
of 
the arrival time of the last request, 
denoted 
by $\hat t_{last}$. We define the error of the prediction, denoted by $\errorlast$, 
to be 
the difference between the predicted and the actual arrival time of the last request: 
$$\errorlast \defeq \hatt_n - t_n.$$
The motivation of 
making such a 
prediction is that the arrival time of the last request provides a guess of the lower bound for the optimal route. 
Later 
we show how to combine this prediction with the \PAH~algorithm and design a learning-augmented algorithm
with consistency, robustness and smoothness. Furthermore, we also show the limit of this prediction, 
i.e., lower bound results.

Note that $\Xpred$ 
reveals 
much more information compared to $\hatt_n$ so that 
the first two prediction models 
would be expected to 
help achieve better 
performance 
than the last model, 
assuming 
the prediction is perfect. 
However, 
when the prediction is not 
good, 
we 
must carefully 
control both $\errortime$ and $\errorpos$ to avoid paying too much 
extra cost.

Here we 
remark 
that very recently Azar et.al.~\cite{azar2022online} proposed a framework for designing learning-augmented online 
algorithms for some graph problems such as Steiner tree/forest, facility location, etc. 
They also 
presented 
a novel definition of prediction errors. However, 
in contrast, 
online routing has to
carefully cope with the predictions of 
arrival time of future requests. Therefore, the notion 
reported 
in~\cite{azar2022online} may not be 
directly applied to 
this study.

\section{Predict a Sequence of Requests}\label{pred:seq}

Given a predicted sequence of requests
$\Xpred$, 
we first discuss the intuition behind the 
two models: 
\emph{prediction without identity} and \emph{prediction with identity}. 
From the perspective of machine learning, 
the former model forecasts by viewing the sequence as a whole, while the latter one makes predictions based on the features of each request individually. Note that both models acquire the entire sequence prediction $\Xpred$ in the beginning. 
Next,
we present two algorithms: \algNI~and \algI, respectively, where  
the former one has
consistency and robustness, and 
the latter one can even additionally
gain smoothness 
by knowing the number of requests.

We remark that, analogous to the discussion in~\cite{ausiello2001algorithms}, 
one can 
disregard the computational complexity of online algorithms and 
assume that 
the server has access to the offline optimal route for a sequence of 
released requests. 
This is 
similar 
to the assumption in~\cite{ausiello2001algorithms} that 
the server is given the optimal solution to visit requests without time constraints. 

\subsection{Sequence Prediction without Identity}\label{sec:3-1}
In this model, the server 
makes 
a prediction $\Xpred$ 
comprising 
$m$ predicted requests.
We present the~\algNI~algorithm, which is $(1.5+\lambda)$-consistent and $(3+2/\lambda)$-robust for $\lambda \in (0, 1]$. 
A na\"{\i}ve idea is that the server finds an 
optimal route $\hat T$ for the prediction $\Xpred$ at time $0$ and 
just 
follows the route $\hat T$ 
%
directly.
It is clearly a 1-consistent algorithm, which, however, 
may 
result in 
an arbitrarily bad robustness.

\begin{restatable}{rThm}{thmlbwo}\label{thm:lb1}
Given a sequence of predicted requests without identity, 
any 1-consistent algorithm has robustness 
of at least $1/\delta$
for any $\delta \in (0,1)$.
\end{restatable}

The brief idea of the~\algNI~algorithm is that the server finds an 
optimal route $\hat T$ for the prediction $\Xpred$ at time $0$ and follows the route 
based on a modified framework of \PAH, which can guarantee a lower bound for 
$\costopt$.
Precisely,
a server operated by \algNI~is given a route $\hat T$ which serves 
the requests in the prediction $\Xpred$ 
with 
the confidence level to the prediction, denoted by $\lambda$, in the beginning. To obtain 
certain 
robustness, 
we set a condition to see whether the current moment $t$ is 
earlier than 
$\lambda |\hat T|$ or not. 
If 
$t < \lambda |\hat T|$, we 
follow
the \PAH~algorithm, unless the route $T_{U_t}$ is too long (i.e., $t + |T_{U_t}| >  \lambda |\hat T|$), where 
$T_{U_t}$ denotes an optimal route that starts serving the set of released unserved requests $U_t$ at time $t$; that is, we adjust the route $T_{U_t}$ by asking the server to 
return home and arrive at the origin $o$ exactly at time $\lambda |\hat T|$.
By 
adding such a gadget, 
we can ensure 
good robustness. 
Otherwise, 
if $ t \ge \lambda |\hat T|$, 
the server 
gets 
to follow the predicted route $\hat T$ 
once 
there are unserved requests (i.e., $U_t \neq \emptyset$). 
Finally, we use the \PAH~algorithm to serve the remaining requests in $\X \setminus \Xpred$ . 
Note that we tend to set the parameter $\lambda$ to a small value if we believe the quality of the prediction is good.

%
In particular, 
next, we show that~\algNI~can still be robust even when the prediction error is not properly defined.
First, we prove it is feasible to add the gadget. 



\begin{restatable}{rLem}{lemtback}\label{lem:tback}
Given that the server follows the~\algNI~algorithm, 
when time $t$ satisfies the condition: $t < \lambda |\hat T|$ and $t + |T_{U_t}| > \lambda |\hat T|$, there exist a 
moment $t_{back}$, $t < t_{back} < t + |T_{U_t}|$, such that $t_{back} + d(p(t_{back}), o) = \lambda |\hat T|$.
\end{restatable}

\begin{algorithm}
    \caption{\textsc{Learning-Augmented Routing Without Identity (LAR-NID)}}
    \label{alg:LASwo}
    \begin{algorithmic}[1] 
    \Require The current time $t$, a sequence prediction $\Xpred$, the confidence level $\lambda \in (0, 1]$, and the set of current released unserved requests $U_t$.
    \State 
    First, 
    compute an optimal route $\hat T$ to serve 
    the requests in $\Xpred$ and return to the origin $o$;
    \State \textbf{While} $t < \lambda |\hat T|$ \textbf{do}
        \State \quad \textbf{If} the server is at the origin $o$ (i.e., $p(t) = o$.) \textbf{then}
            \State \quad \quad Compute an optimal route $T_{U_t}$ to serve all the unserved requests in $U_t$ and return
            \Statex \quad \quad to the origin $o$;
            \State \quad \quad \textbf{If} $t + |T_{U_t}| >  \lambda |\hat T|$ \textbf{then}
            \Comment{Add a gadget}
                \State \quad \quad \quad Find
                the moment 
                $t_{back}$ such that $t_{back} + d(p(t_{back}), o) =  \lambda |\hat T|$;
                \State \quad \quad \quad Redesign a route $T'_{U_t}$ by asking the the server to go back to the origin $o$ at time
                \Statex \quad \quad \quad $t_{back}$ along the shortest path;
                \State \quad \quad \quad Start to follow the route $T'_{U_t}$.
            \State \quad \quad \textbf{Else}
                \State \quad \quad \quad Start to follow the route $T_{U_t}$.
            \State \quad \quad \textbf{EndIf}
        \State \quad \textbf{ElseIf} the server is currently 
        moving along 
        a route $T_{U_{t'}}$, for some $t'< t$ \textbf{then}
            \State \quad \quad \textbf{If} a new request $x_i = (t_i, p_i)$ arrives \textbf{then} 
            \Comment{Similar to PAH}               
            \State \quad \quad \quad \textbf{If} $d(p_i, o) > d(p(t), o)$ \textbf{then}
                    \State \quad \quad \quad \quad Go back to the origin $o$.
                \State \quad \quad \quad \textbf{Else}
                    \State \quad \quad \quad \quad 
                    Move ahead 
                    on the current route $T_{U_{t'}}$.
                \State \quad \quad \quad \textbf{EndIf}
            \State \quad \quad \textbf{EndIf}
        \State \quad \textbf{EndIf}
    \State \textbf{EndWhile}
    
    \State \textbf{While} $t \geq \lambda |\hat T|$ \textbf{then}
        \State \quad Wait until $U_t \neq \emptyset$;
        \State \quad Follow the route $\hat T$ until the server is back to the origin $o$;
        \State \quad Follow \PAH~($t,U_t$).
        \Comment{Serve the remaining requests}
    \State \textbf{EndWhile}
    \end{algorithmic}
\end{algorithm}


\begin{restatable}{rLem}{lemconsistency}\label{lem:consistency}
The~\algNI~algorithm is $(1.5 + \lambda)$-consistent, where $\lambda \in (0,1]$.
\end{restatable}

\begin{restatable}{rLem}{lemrobustness}\label{lem:robustness}
The~\algNI~algorithm is $(3 + 2 / \lambda)$-robust, where $\lambda \in (0,1]$.
\end{restatable}

\begin{restatable}{rThm}{thmLASwo}\label{thm:LASwo}
The~\algNI~algorithm is $(1.5 + \lambda)$-consistent and $(3 + 2 / \lambda)$-robust
but 
not smooth, where $\lambda \in (0,1]$ is the confidence level.
\end{restatable}

\subsection{Sequence Prediction with Identity}\label{sec:3-2}
In contrast to the previous model, we consider having access to the number
of requests, i.e., $n$. 
Given a 
sequence prediction, $\Xpred$ with the size of $n$, 
we first show 
the limitation of this stronger prediction. 

\begin{restatable}{rThm}{thmlbw}\label{thm:lb2}
Given a sequence of predicted requests with identity, 
any $1$-consistent algorithm has robustness at least $2$. 
\end{restatable}

We first consider a 
naive
algorithm, \algT, and
show that it is consistent and smooth but not robust where the details can be founded in the Appendix~\ref{naivealgo}. 
However, we observe that the arrival time of the last request $t_n$ gives a lower bound of the optimal solution, $\costopt$.
As a result, we modify the 
\algT~algorithm and propose the~\algI~algorithm, which is $1$-consistent, $3$-robust and $(2 \errortime + 4 \errorpos)$-smooth.

\begin{algorithm}
    \caption{\textsc{Learning-Augmented Routing Trust (LAR-Trust)}}\label{alg:naive}
    \begin{algorithmic}[1] 
    \Require The current time $t$, the number of requests $n$, a sequence prediction $\Xpred$, and the set of current released unserved requests $U_t$.
    \State 
    Compute an optimal route $\hat T = (\hat x_{(1)}, \ldots, \hat x_{(n)})$ to serve the requests in $\Xpred$ and return to the origin $o$, where $\hat x_{(i)} = (\hat t_{(i)}, \hat p_{(i)})$ 
    denotes
    the $i^{th}$ predicted request 
    in $\hat T$;
    \State Start to follow the route $\hat T$;
    \State \textbf{For} any $i =1, \ldots, n$ \textbf{do}
    
        \State \quad \textbf{If} $t = t_{(i)}$ 
        \textbf{then}
            \State \quad \quad Update the route $\hat T$ by adding the request $x_{(i)}$ after the 
            predicted request $\hat x_{(i)}$;
        \State \quad \textbf{EndIf}
        \State \quad \textbf{If} $p(t) = \hat p_{(i)}$ and $t < t_{(i)}$ \textbf{then}
            \State \quad \quad Wait at position $\hat p_{(i)}$ until time $t_{(i)}$.
            \Comment{Wait until the request arrives}
        \State \quad \textbf{EndIf}
    
    \State \textbf{EndFor}
    \end{algorithmic}
\end{algorithm}

\begin{restatable}{rThm}{thmnaive}\label{thm:naive}
The competitive ratio of the~\algT~algorithm
is $1 + 2 \errortime + 4 \errorpos$.
Thus, the algorithm is $1$-consistent and $(2 \errortime + 4 \errorpos)$-smooth but not robust.
\end{restatable}


The intuition of the~\algI~algorithm is to keep the impact of errors under control. Specifically, the server finds an optimal route $\hat T$ to serve 
the requests in $\Xpred$ in the beginning. Before time $t_n$, the server follows the route $\hat T$ and adjusts it when necessary. To explain how we make adjustments, we describe the route $\hat T$ as a sequence of requests that are in order of priority. That is, $\hat T \defeq (\xhat_{(1)}, \ldots, \xhat_{(n)})$ where $\xhat_{(i)} = (\hat t_{(i)}, \hat p_{(i)})$ 
denotes 
the $i^{th}$ request 
to be served in $\hat T$. If a request $x_i$ arrives, the server modifies the route $\hat T$ by inserting the request $x_{(i)}$ into the sequence $\hat T$ 
after the request $\xhat_{(i)}$.
By updating the route, the server can visit each request $x_i$ with the adjusted version of the route $\hat T$. When the last request $x_n$ arrives (i.e., $t = t_n$), we compute the length of our two possible routes: (1) the remaining distance of $\hat T$ after updates,
denoted 
by
$r_1$, and (2) the distance to go back to the origin and follow the final route $T_{U_{t_n}}$ to visit the requests in $U_{t_n}$,
denoted 
by 
$r_2$. Then, the server chooses the shorter one to visit the remaining requests in $\X$.
For (1), we can rewrite it as the na\"{\i}ve algorithm by updating the route sequentially, where the details can be
founded in the Appendix~\ref{naivealgo}.
We first show that the cost of the~\algT~algorithm
is associated with the error 
defined in Section~\ref{premlim}. Note that the server does not change the order to serve the requests in $\X$ unless the server gives up the route $\hat T$ at time $t_n$.


\begin{algorithm}[!ht]
    \caption{\textsc{Learning-Augmented Routing With Identity (LAR-ID)}}\label{alg:LASw}
    \begin{algorithmic}[1] 
    \Require The current time $t$, the number of requests $n$, a sequence prediction $\Xpred$, and the set of current released unserved requests $U_t$.
    \State $F = 0$;
    \Comment{Initialize $F=0$ to indicate that we trust the prediction} 
    \State 
    First, compute an optimal route $\hat T = (\xhat_{(1)}, \ldots, \xhat_{(n)})$ to serve the requests in $\Xpred$ and return to the origin $o$, where $\xhat_{(i)} = (\hat t_{(i)}, \hat p_{(i)})$ is the $i^{th}$ request to serve in $\hat T$;
    \State Start to follow the route $\hat T$;
    
    \State \textbf{While} $F = 0$ \textbf{do}
    \Comment{Trust the prediction}
        \State \quad \textbf{If} $t = t_{(i)}$ 
        \textbf{then}
            \State \quad \quad $\hat T = (\xhat_{(1)}, \ldots, \xhat_{(i)}, x_{(i)}, \xhat_{(i+1)},  \ldots, \xhat_{(n)})$, for $i = 1, \ldots, n$;
            \Comment{Update the route}
            \State \quad \quad \textbf{If} $t = t_n$ \textbf{then}
            \Comment{Find the shorter route}
                \State \quad \quad \quad $r_1 \leftarrow$ the remaining distance of following $\hat T$;
                \State \quad \quad \quad Compute a route $T_{U_{t_n}}$ to start and finish at the origin $o$ and serve the requests in 
                \Statex \quad \quad \quad $U_{t_n}$;
                \State \quad \quad \quad $r_2 \leftarrow d(p(t), o) + |T_{U_{t_n}}|$;
                \State \quad \quad \quad \textbf{If} $r_1 > r_2$ \textbf{then}
                    \State \quad \quad \quad \quad Go back to the origin $o$;
                    \Comment{Give up the predicted route}
                    \State \quad \quad \quad \quad $F = 1$.
                \State \quad \quad \quad \textbf{Else}
                    \State \quad \quad \quad \quad 
                    Move ahead on the current route $\hat T$.
                \State \quad \quad \quad \textbf{EndIf}
            \State \quad \quad \textbf{Else}
                \State \quad \quad \quad
                Move ahead on the current route $\hat T$.
            \State \quad \quad \textbf{EndIf}
        \State \quad \textbf{EndIf}
        
        \State \quad \textbf{If} $p(t) = \hat p_{(i)}$ and $t < t_{(i)}$ \textbf{then}
            \State \quad \quad Wait at position $\hat p_{(i)}$ until time $t_{(i)}$.
            \Comment{Wait until the request arrives}
        \State \quad \textbf{EndIf}
    \State \textbf{EndWhile}
    \State \textbf{While} $F = 1$ \textbf{do}
    \Comment{Do not trust the prediction}
        \State \quad Start to follow the route $T_{U_{t_n}}$ to serve the requests in $U_{t_n}$.
    \State \textbf{EndWhile}
    \end{algorithmic}
\end{algorithm}

\begin{restatable}{rThm}{thmLASw}\label{thm:LASw}
The competitive ratio of the~\algI~algorithm is $\min \{3, 1 + (2 \errortime + 4 \errorpos) / \costopt\}$. Thus, the algorithm is $1$-consistent, $3$-robust and $(2 \errortime + 4 \errorpos)$-smooth.
\end{restatable}

\begin{restatable}{rCor}{corLASw}\label{cor:LASw}
The~\algI~algorithm is a $\min \{3.5, 2.5 + (3.5 \errortime + 7 \errorpos) / \costopt\}$-competitive polynomial-time algorithm using Christofides' 
heuristic. 
\end{restatable}

\section{Predict the Last Arrival Time}\label{sec:4} 
One can 
observe that Corollary~\ref{cor:PAHwait} provides 
a special 
insight about 
predictions. 
That is, when the last request arrives, 
if the server is at the origin,  
we can apply an offline algorithm 
for the 
TSP to 
serve all 
the remaining unserved requests without dealing with the release time of the requests. Based on the 
insight, 
it is fair to consider 
another 
simpler model by predicting the arrival time of the last requests $\hat{t}_n$ only, 
rather than 
predicting a whole sequence of future requests.  
Though 
we first show 
the restricted power of the model 
due to the limited information of predictions. Next, 
compared to \PAH, 
we introduce a 
pure online algorithm, \redesign, 
and present 
a polynomial-time learning-augmented algorithm 
for the restricted prediction model.

Intuitively, 
if the arrival time of the last request is given, 
the server can just wait until the last request 
arrives before starting a route. 
That is, 
the completion time is $\costalg \leq t_n + |T_{\X}| \leq 2 \costopt$.
Precisely, 
even when the prediction is correct, 
the lower bound of 2, i.e., 
Theorem 3.2 in~\cite{ausiello2001algorithms} still holds and so does the 
corollary.

\begin{theorem}[\cite{ausiello2001algorithms}, Theorem 3.2]\label{thm:PAHlb}
There is no $c$-competitive algorithm for the OLTSP with $c < 2$.
\end{theorem}

\begin{restatable}{rCor}{corlblast}\label{thm:>2}
Given the arrival time of the last request $\hat t_n$, 
there is no $c$-competitive algorithm with $c < 2$.
\end{restatable}

In order to design a learning-augmented online algorithm with robustness, we
present 
a simpler online algorithm, \redesign, 
than \PAH, 
and it can
achieve the same competitive ratio of $3$ in polynomial time. 
Compared to \PAH, 
the~\redesign~algorithm always goes back to the origin and redesign an optimal (or approximate) route for $U_t$, 
whenever a new request $x_i$ arrives. 
The following two lemmas show that the~\redesign~algorithm is $3$-competitive. 

\begin{algorithm}
    \caption{\redesign 
    }\label{alg:replan}
    
    \begin{algorithmic}[1] 
    \Require The current time $t$ and the set of current released unserved requests $U_t$
    \State \textbf{While} $U_t \neq \emptyset$ \textbf{do}
        \State \quad \textbf{If} the server is at the origin $o$ \textbf{then}
            \State \quad \quad Start to follow a route $T_{U_t}$ to serve the requests in $U_t$ and return to the origin $o$.
        \State \quad \textbf{ElseIf} the server is 
        currently moving along
        a route $T_{U_{t'}}$, for some $t' < t$ \textbf{then}
            \State \quad \quad \textbf{If} a new request $x_i = (t_i, p_i)$ arrives \textbf{then}
                \State \quad \quad \quad Go back to the origin $o$.
            \State \quad \quad \textbf{EndIf}
        \State \quad \textbf{EndIf}
    \State \textbf{EndWhile}
    \end{algorithmic}
\end{algorithm}

\begin{restatable}{rLem}{lemhalf}\label{lem:1/2}
Given 
that $p(t)$ is the current position of a server 
operated by the~\redesign~algorithm, we have $d(p(t), o) \leq \frac{1}{2} \costopt$, for any $t$.
\end{restatable}

\begin{restatable}{rLem}{lemreplan}\label{lem:replan}
The~\redesign~algorithm is
a $3$-competitive polynomial-time algorithm for the OLTSP in a metric space using Christofides' 
heuristic. 
\end{restatable}

\subsection{The Algorithm}\label{sec:4-1}
In this section, we present the~\algL~algorithm and show that it is $\min \{4, 2.5 + |\errorlast| / \costopt \}$-competitive. The idea of the algorithm is to ensure the server 
to arrive 
at the origin $o$ at time $\hatt_n$ and follow the~\redesign~algorithm,  
for the rest of the 
execution. 

It is motivated by the desire to 
enable 
the server 
to approach
the origin at time $t_n$, when the last request arrives.
An intuitive strategy is to wait at the origin $o$ until the predicted last arrival time $\hat t_n$ and then follow a route $T_{U_{\hat t_n}}$ that serves whatever is in hands. It seems to beat the previous result easily when the prediction is perfectly accurate, i.e., $\hat t_n = t_n$. However, waiting at the origin is potentially costly when the quality of predictions is undisclosed. If the expected time $\hat t_n$ is too late, the server might postpone serving requests for a long time. 
In consequence, the competitive ratio might extend to infinity 
so that 
the algorithm has no robustness. 

We propose a learning-augmented algorithm, 
called~\algL.
Basically, 
we incorporate the predicted time $\hat t_n$ into the~\redesign~algorithm by forcing the server to return 
to the origin at time $\hatt_n$ in order to hedge against the possible loss.
At any time $t$, if the server is at the origin $o$, it finds an approximate route $T_{U_t}$ to visit all released unserved requests in $U_t$ by using Christofides' 
heuristic. 
Similar to~\algNI, we add a gadget here: if the route is too long (i.e., $t < \hat t_n < t+|T_{U_t}|$), the server should find the moment $t_{back}$ to stop the route $T_{U_t}$ and start moving towards the origin 
so that the server can arrive 
at the origin 
exactly at time $\hat t_n$. 
On the other hand, if the server is 
moving along 
a route $T_{U_{t'}}$ and a new request $x_i$ arrives,
the server goes back to the origin $o$ immediately at time $t_i$, as~\redesign~performs.

\begin{algorithm}
    \caption{\textsc{Learning-Augmented Routing With Last Arrival
    Time (LAR-LAST)}}\label{alg:LAM}
    \begin{algorithmic}[1] 
    \Require The current time $t$, the predicted last arrival time $\hat t_n$, and the set of current released unserved requests $U_t$.
    \State \textbf{While} $U_t \neq \emptyset$ \textbf{do}
    \State \quad \textbf{If} the server is at the origin $o$ (i.e., $p(t) = o$) \textbf{then}
        \State \quad \quad Compute an approximate route $T_{U_{t}}$ to serve all the requests in $U_t$ and return to the  
        \Statex \quad \quad origin $o$;
        \State \quad \quad \textbf{If} $t < \hat t_n$ and $t + |T_{U_{t}}| > \hat t_n$ \textbf{then}
        \Comment{Add a gadget}
            \State \quad \quad \quad Find 
            the moment $t_{back}$ such that $t_{back} + d(p(t_{back}), o) = \hat t_n$;
            \State \quad \quad \quad Redesign a route $T'_{U_{t}}$ by asking the the server to go back to the origin $o$ at time
            \Statex \quad \quad \quad $t_{back}$ along the shortest path;
            \State \quad \quad \quad Start to follow the route $T'_{U_{t}}$.
        \State \quad \quad \textbf{Else}
            \State \quad \quad \quad Start to follow the route $T_{U_{t}}$.
        \State \quad \quad \textbf{EndIf}
    \State \quad \textbf{ElseIf} the server is currently 
    moving along a route $T_{U_{t'}}$, for some $t' < t$ \textbf{then}
        \State \quad \quad \textbf{If} a new request $x_i = (t_i, p_i)$ arrives \textbf{then}
        \Comment{Similar to \redesign}
            \State \quad \quad \quad Go back to the origin $o$.
        \State \quad \quad \textbf{EndIf}
    \State \quad \textbf{EndIf}
    \State \textbf{EndWhile}
    \end{algorithmic}
\end{algorithm}

Next, we prove the main result of the~\algL~algorithm. Note that the existence of $t_{back}$ can be 
similarly 
derived from Lemma~\ref{lem:tback}.

\begin{restatable}{rThm}{thmLAM}\label{thm:LAM}
The~\algL~algorithm is a $\min \{ 4, 2.5 + |\errorlast| / \costopt \}$-competitive polynomial-time algorithm,
where $\errorlast = \hat t_n - t_n$. Therefore, the~\algL~algorithm is $2.5$-consistent, $4$-robust and $|\errorlast|$-smooth.
\end{restatable}
\section{Extension to The Dial-a-Ride Problem}
As we found the OLDARP has some properties similar to the OLTSP, we extend our models and algorithms to the dial-a-ride problem with unlimited capacity. 
We only mention the necessary changes for adapting the proposed algorithms to this problem as well as the results we obtain, and the details are presented in the Appendix~\ref{LADARP}. The main difference is that the learning-augmented algorithms for the OLDARP are based on the~\redesign~algorithm,
which yields a competitive ratio of $2.5$ when ignoring the computational issue. Note that although the \algsmartstart~algorithm~\cite{ascheuer2000online} has a better competitive ratio than the \redesign~strategy,
it is also more complicated and thus more difficult to be integrated with predictions.

First, we 
develop the~\algdarNI~algorithm by 
substituting the \redesign~strategy for the \PAH~algorithm in~\algNI. 
Note that if the server 
gets to go back to the origin when carrying some requests, it brings all of them to the origin. This algorithm is $(1.5+ \lambda)$-consistent and $(3.5 + 2.5 / \lambda)$-robust but not smooth.
Next, we consider the~\algdarI~algorithm. 
Note that a route for the OLDARP can be described as a series of pickup and delivery positions instead of requests. 
Accordingly, when a request becomes known, the server needs to update two positions.
Then, the algorithm is 
$1$-consistent, $3$-robust, and $(2 \errortime + 4 \errorpos)$-smooth.
Last, we use the~\algdarL~algorithm for the prediction of last arrival time, which replaces Christofides' heuristic with the optimal solutions
for the offline DARP without release time. This change makes the algorithm
$2$-consistent, $3.5$-robust, and $|\errorlast|$-smooth.
\section{Conclusion}
In this study, we have investigated 
two well-known online routing problems 
based on the learning-augmented framework. 
The proposed prediction models and results raise some interesting 
questions to further explore. First, it 
would be of interest to improve the current 
competitive results;  
especially coping with waiting strategies is usually helpful to 
the OLTSP and the OLDARP. 
Note that the currently best results for 
these two problems
exploited waiting strategies, which leaves us an open problem to 
improve our learning-augmented online algorithms by incorporating a waiting strategy.
It would be also worthwhile to design learning-augmented online algorithms for other variants of online routing problems 
with
limited capacity or 
multiple servers. 
\newpage
\nocite{*}
\bibliography{references}
\bibliographystyle{plain}

\newpage

\appendix

\section{Detail of the \PAH~Algorithm}\label{PAH:code}
Here we present the details of \PAH.
\begin{algorithm}
\caption{\PAH
}\label{alg:PAH}
    \begin{algorithmic}[1] 
    \Require The current time $t$ and the set of current released unserved requests $U_t$
    \State \textbf{While} $U_t \neq \emptyset$ \textbf{do}
    \State \quad \textbf{If} the server is at the origin (i.e., $p(t) = o$) \textbf{then}
        \State \quad \quad Start to follow an optimal (or approximate) route $T_{U_t}$ 
        passing through each 
        request 
        \Statex \quad \quad in $U_t$.
    \State \quad \textbf{ElseIf} the server is 
    currently moving along a route $T_{U_{t'}}$, for some $t' < t$ \textbf{then}
        \State \quad \quad \textbf{If} a request $x_i = (t_i, p_i)$ arrives \textbf{then}
            \State \quad \quad \quad \textbf{If} $d(p_i, o) > d(p(t), o)$ \textbf{then}
                \State \quad \quad \quad \quad Go back to the origin $o$.
            \State \quad \quad \quad \textbf{Else}
                \State \quad \quad \quad \quad
                Move ahead
                on the current route $T_{U_{t'}}$.
            \State \quad \quad \quad \textbf{EndIf}
        \State \quad \quad \textbf{EndIf}
    \State \quad \textbf{EndIf}
    \State \textbf{EndWhile}
    \end{algorithmic}
\end{algorithm}

\section{Missing Proofs and Details in Section~\ref{pred:seq}}


\subsection{Missing Proofs in Section~\ref{sec:3-1}}
\thmlbwo*
\begin{proof}
Consider the 
OLTSP on the real line, a special metric space. 
Let the prediction be $\Xpred = \{x_1 = (\delta, \delta), x_2 = (1, 1)\}$, where $\delta \in (0, 1)$.
First, 
suppose 
the actual input is $\X = \Xpred$. In this case, the
offline
optimal solution is to start moving at time $0$, visit $x_1$ at time $\delta$ and $x_2$ at time $1$, and return to the origin at time $2$. To obtain the consistency of $1$, the server must follow the exact same route as the optimal one.
On the other hand, 
if the 
actual input is $\X' = \{x'_1 = (\delta, \delta)\}$, 
i.e. only one request. 
Since the number of requests is unknown, the algorithm recognizes the error at time $1$ and cannot 
go back home until 
time $2$. However, the offline optimal 
solution 
finishes at time $2 \delta$. Thus, the competitive ratio is $\costalg / \costopt \geq 2 / (2 \delta) = 1 / \delta$.
\end{proof}

\lemtback*
\begin{proof}
Define the function $f: [t, t + |T_{U_t}|] \rightarrow [t, t + |T_{U_t}|]$ 
to be 
the moment 
when the server arrives at the origin again; 
that is, 
the server stops following route $T_{U_t}$ and goes back to the origin 
at a particular moment $t_{back}$, i.e., $f(t_{back}) = t_{back} + d(p(t_{back}), o)$. 
Then, 
we let the function $g(x) \defeq f(x)-\lambda |\hat T|$, 
and it has two properties: $g(t) = t - \lambda |\hat T| < 0$ and $g(t + |T_{U_t}|)= t + |T_{U_t}| -\lambda |\hat T| > 0$. Since $g$ is continuous, there must be at least one point $t_{back}\in (t, t + |T_{U_t}|)$ such that $g(t_{back})= 0$ and 
thus $t_{back} + d(p(t_{back}), o) = \lambda |\hat T|$. Without loss of generality, 
we choose the last one 
satisfying $g(t_{back})= 0$.
\end{proof}

\lemconsistency*

\begin{proof}
The server 
gets 
to follow the 
modified framework 
of the \PAH~algorithm at time $0$ and returns to the origin at 
exactly 
time $\lambda |\hat T|$. Then, at time $\lambda |\hat T|$, the server 
follows the predicted route $\hat T$
if  $U_{\lambda |\hat T|} \neq \emptyset$; 
otherwise, it waits until the next request arrives. 
Assume the prediction is perfect, i.e., $\Xpred = \X$ and $|\hat T| = \costopt$.
Since $\hat T$ is an optimal route, the server completes after following $\hat T$.
We consider the following two cases:

\begin{itemize}
    \item {1. $U_{\lambda |\hat T|} \neq \emptyset$.} In this case, the server gets to follow route $\hat T$ at time $\lambda |\hat T|$, and the completion time of the algorithm is $\costalg = \lambda |\hat T| + |\hat T| \leq (1 + \lambda) \costopt$. 
    \item {2. $U_{\lambda |\hat T|} = \emptyset$.} In this case, after time $\lambda |\hat T|$, the server waits for the next request to 
    move. 
    Let $x_j = (t_j, p_j)$ denote this request, which implies 
    $t_j = \min_t \{t > \lambda |\hat T| : U_t \neq \emptyset \}$. The completion time of the algorithm is $\costalg = \lambda |\hat T| + (t_j - \lambda |\hat T|) + |\hat T|$. Note that $t_j - \lambda |\hat T|$ is the time length
    that the server waits at the origin. 
    We split the duration $t_j - \lambda |\hat T|$ into two time intervals according to the operations of the offline optimal server, 
    i.e., OPT: 
    $t_{move}$ and $t_{idle}$. 
    The former denotes how much time 
    the OPT
    is moving, and 
    the latter
    denotes the time length 
    the 
    OPT
    just waits at the origin. 
    That is, 
    $t_{move} + t_{idle} = t_j - \lambda |\hat T|$. Note that 
    \algNI~is 
    not visiting any request during $t_{idle}$, 
    either. Because 
    the case $t_{idle} \neq 0$ favors the algorithm, 
    we simply 
    consider
    the worst case $t_{idle} = 0$ and thus $t_{move} = t_j - \lambda |\hat T|$. Intuitively, the 
    OPT
    gains the advantage by moving toward the next request 
    during $t_{move}$. 
    Moreover, 
    the time/distance 
    the OPT
    saves by moving earlier 
    is at most 
    $\frac{1}{2} \costopt$ 
    due to 
    the triangle inequality. 
    Since the server moves at unit speed, we have $t_{move} \leq \frac{1}{2} \costopt$. To sum up, the completion time is $\costalg \leq \lambda |\hat T| + t_{move} + |\hat T| \leq (1.5 + \lambda) \costopt$.
\end{itemize}
\end{proof}

\lemrobustness*
\begin{proof}
Note that $t_i \leq Z^{OPT}$ for any $i$ and $|T_{U_{t}}| \leq Z^{OPT}$ for any $t$.
To show the robustness, we consider the following cases: 

\begin{itemize}
    \item {1. Assume the server finishes before $\lambda |\hat T|$.} In this case, the server acts exactly like the \PAH~algorithm and does not follow route $\hat T$. 
    Thus, by Theorem~\ref{thm:PAH2}, the server obtains the bound $\costalg \leq 2 \costopt$.
    \item {2. The server cannot finish before $\lambda |\hat T|$.} In this case, there are some unserved requests at time $\lambda |\hat T|$ or some request 
    arriving after time $\lambda |\hat T|$. In both cases, the algorithm first follows $|\hat T|$ and then switches to the \PAH~algorithm if there 
    remain some unserved requests. Also, since the algorithm cannot serve all of the requests before time $\lambda |\hat T|$, it implies that $\costopt \ge \lambda |\hat T|/2$; otherwise, 
    it contradicts Theorem~\ref{thm:PAH2}.
    To analyze the performance of 
    \algNI~in this case, we 
    divide
    the case into the following four subcases: 
        \begin{itemize}
            \item{2-1. $U_{\lambda |\hat T|} \neq \emptyset$ and $t_n > (1 + \lambda)|\hat T|$.} 
            Since the~\algNI~algorithm follows \PAH~after time $(1 + \lambda)|\hat T|$, we have $\costalg \leq 2 \costopt$ by Corollary~\ref{cor:PAHwait}.
            \item{2-2. $U_{\lambda |\hat T|} \neq \emptyset$ and $t_n \leq (1 + \lambda)|\hat T|$.} 
            By using $\lambda |\hat T| / 2 \leq Z^{OPT}$,
            the completion time of the algorithm can be 
            derived 
            as follows: 
            \begin{align*}
                \costalg & \leq \lambda |\hat T| + |\hat T| + |T_{U_{(1+\lambda)|\hat T|}}|
                \\& \leq 2 \costopt + (2 / \lambda) \costopt + \costopt
                \\& \leq (3 + 2 / \lambda) \costopt
            \end{align*}
            \item{2-3. $U_{\lambda |\hat T|} = \emptyset$  and $t_n > t_j + |\hat T|$.} This case is similar to the case 2-1. 
            The~\algNI~algorithm follows \PAH~after time $t_j + |\hat T|$. Thus, we have $\costalg \leq 2 \costopt$ by Corollary~\ref{cor:PAHwait}.
            \item{2-4. $U_{\lambda |\hat T|} = \emptyset$  and $t_n \leq t_j + |\hat T|$.} 
            Some request arrives when the server is on the route $\hat T$, which implies that $\lambda |\hat T| < t_n \leq Z^{OPT}$. Also, the server starts its 
            last route at time $t_j + |\hat T|$, when it is back 
            home. 
            The completion time of~\algNI~can be 
            obtained 
            as follows: 
            \begin{align*}
                \costalg & \leq t_j + |\hat T| + |T_{U_{t_j + |\hat T|}}|
                \\& \leq \costopt + (1 / \lambda) \costopt + \costopt
                \\& \leq (2 + 1 / \lambda) \costopt
            \end{align*}
        \end{itemize}
\end{itemize}
The proof is complete. 
\end{proof}


\thmLASwo*
\subsection{Missing Proofs and Details in Section~\ref{sec:3-2}}
\thmlbw*
\begin{proof}
Consider the 
OLTSP on the real line, 
similar to 
the proof of Theorem~\ref{thm:lb1}, and assume the prediction is $\Xpred = \{ x_1 = (0.5, 0.5), x_2 = (1, 1)\}$. On one hand, 
suppose 
the prediction is perfect, i.e., $\X = \Xpred$. The 
OPT gets to 
serve $x_2$ at time $1$ and returns to the origin at time $2$. 
In order to reach the consistency of $1$, the server must follow this route. On the other hand, we 
assume 
the actual input is $\X' = \{x'_1 = (0.5, 0.5), x'_2 = (1, 0)\}$. As the server can recognize the 
prediction error once upon the arrival of 
the second request $x_2$, 
it is at position $1$ at time $1$ and cannot 
return 
to the origin 
until 
time $2$. However, the 
completion time 
of the OPT 
is $\costopt = 1$. Thus, the competitive ratio is $\costalg / \costopt = 2$.
\end{proof}

\noindent{\textbf{The~\algT~Algorithm}}\label{naivealgo}
We first recall the~\algT~algorithm that the server follows a route for the prediction and only adjust the predicted route when an error is discovered. Formally, the server gets an optimal route $\hat T$ to visit the requests in $\Xpred$ in the beginning and starts following $\hat T$ immediately. As the requests arrive, the server modifies the route $\hat T$ in two ways: (1) if the server moves to a request's position $\hat p_i$ but the request $x_i$ has not arrived, the server waits for it to 
move; 
(2) if a request $x_i$ arrives, the server adjusts the route $\hat T$ by inserting the actual request $x_i$ after the corresponding predicted one $\hat x_i$. In this case, the server can visit all requests in $\X$ by following the adjusted route 
of $\hat T$. Note that $\hat T$ decides the order 
when 
the server visits these requests, and the server 
does not change the order even when the 
prediction errors are large. We analyze the performance of
the~\algT~algorithm in the following paragraph.
To analyze the performance of the~\algT~algorithm, we 
show that the error term can be used to bound the cost of two different routes. Then, to compare the cost of two different routes, we first consider a restricted condition where the position of each request of the routes is the same but with a different arrival time.

\thmnaive*

\begin{proof}
Let $T_\X$ be the 
optimal route for the requests in $\X$ and $T_{\Xpred}$ be the 
optimal route to visit the requests in $\Xpred$. As mentioned in Section~\ref{sec:3-2}, a route can be regarded as the order to serve the requests. We consider the total completion time of the following routes: 
(1) let $\costopt$ be the completion time of $T_\X$,
(2) let $Z^*_{\Xpred}$ be the completion time if the server visits the requests in $\Xpred$ by following the order of the route $T_\X$, 
(3) let $Z'_{\Xpred}$ be the completion time of $T_{\Xpred}$, and
(4) let $\costalg$ be the completion time of \algT, which visits the requests in $\X$ by following the order of the route $T_{\Xpred}$.
To 
compare 
$\costalg$ 
with 
$\costopt$, the proof is split into the following two parts: we first show that the completion time $Z'_{\Xpred}$ of the predicted route $\hat T$ can be 
bounded in terms of 
$\costopt$, and then 
we  
bound $\costalg$ 
using 
$Z'_{\Xpred}$. Finally, we obtain the desired result by combining the results of the 
two parts. 

\begin{lemma}\label{clm:1}
$Z^{'}_{\Xpred} \leq Z^*_{\Xpred} \leq \costopt + \errortime + 2 \errorpos$. 
\end{lemma}
\begin{proof}
Consider the 
route $Z^*_{\Xpred}$: 
(1) the server waits at the origin $o$ for time $\errortime$ before following the route $T_\X$, and (2) after the server arrives 
at $p_i$, it moves to $\hat p_i$ and goes back to $p_i$ to continue the route $T_\X$. 
We argue that the server can visit all of requests in $\Xpred$ and $\X$ then return to the origin at time $\costopt + \errortime + 2 \errorpos$. 
For (1), it guarantees that the server does not have to pay any extra waiting cost compared to the 
offline 
optimal route 
plus the waiting time $\errortime$. For (2), it takes $2 d(p_i, \hat p_i)$ to serve request $\hat x_i$ then go back to $x_i$. Combing (1) and (2)
leads to 
that the server can visit all of requests in $\Xpred$ and $\X$ then return to the origin at time $Z^*_{\Xpred} \leq \costopt + \errortime + 2 \errorpos$. Finally, since $T_{\Xpred}$ is the optimal route for $\Xpred$, we 
thus 
have $Z^{'}_{\Xpred} \leq Z^*_{\Xpred} $.
Note that the result implies that 
compared with $T_\X$, the server might find a longer route $\hat T$ if 
there is some request 
expected to arrive later, 
or if 
the prediction error in positions exists.
\end{proof}

\begin{lemma}\label{clm:2}
$\costalg \leq Z'_{\Xpred} + \errortime + 2 \errorpos$. 
\end{lemma}
\begin{proof}
Consider the following route 
$Z'_{\Xpred}$: 
(1) the server waits at the origin $o$ for time $\errortime$ before following the route $\hat T$, and (2) after the server arrives $\hat p_i$, it moves to $p_i$ and goes back to $\hat p_i$ to continue the route $\hat T$. We argue that the server can visit all of requests in $\Xpred$ and $\X$, 
and then returns to the origin at time $Z'_{\Xpred} + \errortime + 2 \errorpos$. For (1), 
it ensures that $x_i$ has arrived 
at $p_i$ when the server is at $\hat p_i$. For (2), it takes $2 d(\hat p_i, p_i)$ to serve request $x_i$ then go back to $\hat x_i$, which is exactly the operation of the \algT~algorithm. Therefore, we can derive that the server can visit all of requests in $\Xpred$ and $\X$,  
and 
then returns to the origin at time $Z^{ALG} \leq Z'_{\Xpred} + \errortime + 2 \errorpos$. 
Notice 
that the server might need to adjust the route $\hat T$ if some request arrives later than expected, or if 
some 
position error exists.

\end{proof}

Finally, combining the results of Lemma~\ref{clm:1} and Lemma~\ref{clm:2} 
can obtain

\begin{align*}
    \costalg \leq \costopt + 2 \errortime + 4 \errorpos.
\end{align*}

\end{proof}

\thmLASw*

\begin{proof}
Initially, the server follows an optimal route $\hat T$ for the requests in $\Xpred$,
which implies it is 1-consistent when the predictions are 
accurate, 
i.e. $\X = \Xpred$. 
The server makes adjustments if needed (i.e.,  $p_i \neq \hat p_i$ for some $i$) before and when the last request arrives. 
Finally, at time $t_n$, it decides 
whether the server 
continue on the current route or not. Therefore, we
discuss the following two cases depending on the relationship of $r_1$ and $r_2$. Note that $t_n \leq \costopt$ and 
$|T_{U_{t}}| \leq Z^{OPT}$ for any $t$.
\begin{itemize}
    \item{1. $r_1 \leq r_2$.} In this case, the server knows that continuing on the route $\hatT$ takes less time than redesigning a new one, which implies the errors are small. Thus, it follows the route $\hat T$ 
    from the beginning to the end. By Theorem~\ref{thm:naive}, the completion time of the algorithm can be bounded by $\costalg \leq \costopt + 2 \errortime + 4 \errorpos$.
    \item{2. $r_1 > r_2$.} In this case, it is better to go back to the origin and design a new route. The server 
    returns to 
    the origin $o$ at time $t_n + d(p(t_n), o)$ and starts to follow a new route $T_{U_{t_n}}$, which visits all unserved requests in 
    $U_{t_n}$. Because there is no 
    more request, 
    the server completes all of the requests at time $\costalg = t_n + d(p(t_n), o) + |T_{U_{t_n}}|$. In addition, since the server moves with unit speed, 
    $d(p(t_n), o)= d(o, p(t_n)) \leq t_n$. Combining 
    the two 
    inequalities 
    leads to 
    the result $\costalg \leq 3 \costopt$.
\end{itemize}
\smallskip
As the server chooses the 
faster 
route among the two possible options, the competitive ratio can be stated as $\min \{3, 1+(2 \errortime + 4 \errorpos) / \costopt\}$.
\end{proof}



\corLASw*
\begin{proof}
We first consider the case 
that 
the server follows the optimal route $\hat T$ 
from the beginning to the end.
Let $T_\X$ be the optimal route for the requests in $\X$ and $T_{\Xpred}$ be the approximate route to visit the requests in $\Xpred$. We consider the total completion time of the following routes: 
(1) let $\costopt$ be the completion time of $T_\X$,
(2) let $Z^*_{\Xpred}$ be the completion time if the server visits the requests in $\Xpred$ by following the order of the route $T_\X$, 
(3) let $Z'_{\Xpred}$ be the completion time of $T_{\Xpred}$, and
(4) let $\costalg$ be the completion time of \algT, which visits the requests in $\X$ by following the order of the route $T_{\Xpred}$.

\begin{lemma}\label{lem:polynaive}
$\costalg \leq 2.5 \costopt + 3.5 \errortime + 7 \errorpos$.
\end{lemma}
\begin{proof}
We split into the following three parts:
\begin{item}
    \item{1. $Z^*_{\Xpred} \leq \costopt + \errortime + 2 \errorpos$.} Consider the route $Z^*_{\Xpred}$ such that the server waits at the origin $o$ for time $\errortime$ before following the route $T_{\X}$ and moves from $p_i$ to $\hat p_i$ to serve requests in $\Xpred$. The completion time of the route is $Z^*_{\Xpred} \leq \costopt + \errortime + 2 \errorpos$. 
    \item{2. $Z'_{\Xpred} \leq 2.5 Z^*_{\Xpred}$.} Given that Christofides' algorithm is $1.5$-approximate, we finds an upper bound for the approximation ratio of the offline TSP with release time; assuming the server waits at the origin $o$ until the last requests arrive, i.e., $t = t_n$, and then follows an approximate route $T$ found by Christofides' algorithm, the server can complete serving all the requests at time $t_n + |T| \leq 2.5 Z^{OPT}$. That is, the completion time of the approximate route $Z'_{\Xpred}$ can be bounded by $2.5$ times the completion time of any other route for the requests in $\Xpred$.
    \item{3. $\costalg \leq Z'_{\Xpred} + \errortime + 2\errorpos$.} Note that the route $\hat T$ can visit all of the requests in $\Xpred$. Consider the route that the server waits at the origin $o$ for $\errortime$ and moves from $\hat p_i$ to $p_i$ to visit request $x_i$. The server can complete serving the requests in $\X$ at time $\costopt + \errortime + 2 \errorpos$.
\end{item}
Combining the above three inequalities, the result is as follows:
    \begin{align*}
        \costalg &\leq Z'_{\Xpred} + \errortime + 2 \errorpos
        \\& \leq 2.5 Z^*_{\Xpred} + \errortime + 2 \errorpos
        \\& \leq 2.5 (\costopt + \errortime + 2 \errorpos) + \errortime + 2 \errorpos
        \\& \leq 2.5 \costopt + 3.5 \errortime + 7 \errorpos
    \end{align*}
\end{proof}
Then, we discuss the two following routes:
\begin{itemize}
    \item{1. $r_1$.} The server follows an optimal route $T_{\hat X}$ to serve the requests in $\X$. By Lemma~\ref{lem:polynaive}, the completion time is bounded by $\costalg \leq 2.5 \costopt + 3.5 \errortime + 7 \errorpos$.
    \item{2. $r_2$.} The server goes back to the origin $o$ and starts its final route at time $t_n$. The completion time is $Z^{ALG} = t_n + d(p(t_n), o) + |T_{U_{t_n}}|$. By $|T_{U_{t_n}}| \leq 1.5 Z^{OPT}$, we have $Z^{ALG} \leq 3.5 Z^{OPT}$.
\end{itemize}
As the server chooses the fastest route, the competitive ratio is $\min \{3.5, 2.5 + (3.5 \errortime + 7 \errorpos) / Z^{OPT}\}$.

\end{proof}


\section{Missing proofs in Section~\ref{sec:4}}
\lemhalf*
\begin{proof}
A simple explanation is that 
(1) at any time $t$ during the execution of the algorithm, the distance between the online server and the origin, i.e. $d(p(t), o)$ is at most the length between the farthest request and the origin, and 
(2) the offline 
optimal route must be longer than twice the distance 
between the farthest request and the origin.
Now, we prove the lemma in a more rigorous way by specifying the state of the server at time $t$.

\begin{itemize}
    \item {1. When the server is at the origin $o$ at time $t$,} trivially, we get $d(p(t), o) = 0 \leq \frac{1}{2} \costopt$.
    \item {2. Otherwise, the server is not at the origin $o$ at time $t$.} 
    In this case, we know that the server is either on a route $T_{U_{t'}}$ found by Christofides' algorithm, where $t' < t$, or already on its way home due to the arrival of some earlier request. 
    Let $T^*$ denote an optimal route to visit the set of requests for which the route $T_{U_{t'}}$ is planned.
    Note that neither route $T_{U_{t'}}$ nor route $T^*$ considers the release time of the requests.
    \begin{itemize}
        \item {2-1. The server is on the route $T_{U_{t'}}$.} This implies that the server is traveling between a pair of points, denoted as $a$ and $b$, where a point must be the position of a request $x_i$ or the origin $o$. Since route $T^*$ contains at least point $a$, point $b$, and the origin $o$, we know $|T^*| \geq d(o, a) + d(a, b) + d(b, o)$. Then, we have $d(p(t), o) \leq \min \{ d(p(t), a) + d(a, o), d(p(t),b) + d(b, o)\} \leq \frac{1}{2} |T^*| \leq \frac{1}{2} \costopt$. 
        \item {2-2. The server is on its way back home.} 
        This implies that the server has terminated a route $T_{U_{t'}}$.
        We know that some request $x_i=(t_i, p_i)$ arrives between the time $t'$
        and time $t$. 
        As we can obtain $d(p(t_i), o) \leq \frac{1}{2} \costopt$ by the above argument, we 
        also have 
        $d(p(t), o) \leq d(p(t_i), o) \leq \frac{1}{2} \costopt$.
    \end{itemize}
\end{itemize}
\end{proof}

\lemreplan*

\begin{proof}
Note that the algorithm 
finishes 
when the server receives the last request $x_n$, goes back to the origin $o$, and follows the last route $T_{U_{t_n}}$ found by Christofides' algorithm.
Namely, the completion time is $\costalg = t_n + d(p(t_n), o) + |T_{U_{t_n}}|$.
By the two natural bounds $t_n \leq \costopt$ and $|T_{U_{t_n}}| \leq 1.5 \costopt$, we know the completion time is bounded by $\costalg \leq 2.5 \costopt + d(p(t_n), o)$. Meanwhile, we have $d(p(t_n), o) \leq 0.5 \costopt$ by Lemma~\ref{lem:1/2}. Combining them, we get $\costalg \leq 3 \costopt$.
\end{proof}


\subsection{Missing Proofs in Section~\ref{sec:4-1}}
\thmLAM*

\begin{proof}
The most important thing is the server's position when the last request $x_n$ arrives, and the ideal situation is that the server is at the origin $o$ at time $t_n$, i.e., $p(t_n) = o$. We distinguish between two main cases depending on the relationship between the predicted and actual last arrival time, $\hat t_n$ and $t_n$. 
Note that we get $t_n \leq \costopt$ and $|T_{U_t}| \leq 1.5 \costopt$ for any $t$ trivially.

\begin{itemize}
    \item {1. $\hat t_n \leq t_n$.} The server is at the origin $o$ at time $\hatt_n$. After then, it receives some request and returns to the origin only when a request arrives or a route is completed. Therefore, the completion time is $\costalg = t_n + d(p(t_n), o) + |T_{U_{t_n}}|$. By the above two inequalities, this leaves us $\costalg \leq 2.5 \costopt + d(p(t_n), o)$ to discuss.
        \begin{itemize}
            \item {1-1.} By Lemma~\ref{lem:1/2}, the distance from the server to the origin is bounded at any moment. We get $d(p(t_n), o) \leq 0.5 \costopt$ and our first bound $\costalg \leq 3 \costopt$. 
            \item {1-2.} Although the server might have left the origin at time $t_n$, it cannot be too far if the error $|\errorlast|$ is small. Seeing that the server moves with unit speed, the distance it can travel between time $\hat t_n$ and $t_n$ is at most $t_n - \hat t_n$. As the server is at the origin $o$ at time $\hat t_n$, we obtain $d(p(t_n), o) \leq t_n - \hat t_n = - \errorlast$. Thus, our second bound for case 1 is $\costalg \leq 2.5 \costopt + (- \errorlast)$, where $\errorlast \leq 0$.
        \end{itemize}
    \item {2. $\hat t_n > t_n$.} 
    Since no request arrives after time $\hat t_n$ and the server could start its last route $T_{U_{\hat t_n}}$ at $\hat t_n$ (if needed), we have $\costalg \leq \hatt_n + |T_{U_{\hat t_n}}|$.
    Let $t_L$ be the last time before $\hat t_n$ that a route is planned and $T^L$ be the original route planned at time $t_L$. Let $T'^L$ be the new route if $T^L$ is adjusted and equal to $T^L$ if not. Note that $t_L + |T'^L| \leq \hat t_n$ due to the choice of $t_{back}$.
        \begin{itemize}
            \item {2-1. $t_L + |T^L| \leq \hat t_n$.} The server finishes serving all requests in $\X$ before time $\hat t_n$ and thus $\costalg \leq \hatt_n$. Note that none of the routes the server has followed is too long and needs to be adjusted. In this case, since the algorithm runs in the same way as the~\redesign~algorithm, we know that $\costalg \leq 3 \costopt$ still holds by Lemma~\ref{lem:replan}.
            \item {2-2. $t_L + |T^L| > \hat t_n$ and $t_n \leq t_L$.} The route $T^L$ is too long, and the algorithm find a time $t_{back}$ and a substitute route $T'^L$ so that the server can be at the origin $o$ at time $\hat t_n$. Since every request in $\X$ has arrived before $t_L$, the server could have finished visiting all requests by following the original route $T^L$. However, the server chooses the adjusted route $T'^L$ instead. The additional cost of moving back to the origin is at most $2 d(p(t_{back}), o)$. Thus, by Lemma~\ref{lem:1/2} and Lemma~\ref{lem:replan}, the completion time is $\costalg \leq 3 \costopt + 2 d(p(t_{back}), o) \leq 4 \costopt$.
            \item {2-3. $t_L + |T^L| > \hat t_n$ and $t_n > t_L$.} Again, the route $T^L$ is too long, and at least one request is unknown at time $t_L$.
            Also, the server cannot visit all requests in $\X$ before $\hat t_n$ and has to start a new route at that moment, which is the last route since all requests have arrived by then.  Noting the route $T'^L$ also satisfies $|T'^L| \leq 1.5 \costopt$, the completion time is 
            \begin{align*}
                \costalg &\leq \hatt_n + |T_{U_{\hat t_n}}|
                \\& = t_L + |T'^L| + |T_{U_{\hat t_n}}|
                \\& \leq t_n + |T'^L| + |T_{U_{\hat t_n}}|
                \\& \leq 4 \costopt
            \end{align*}
            Combining with the results in 2-1 and 2-2, we have our first bound $\costalg \leq 4  \costopt$.
            \item {2-4.} If the error $|\errorlast|$ is not zero, the server might have not reached the origin at time $t_n$. However, if the error is small, it must be close to the origin since it has planned to reach there at time $\hat t_n$. 
            Formally, we have $\costalg \leq t_n + \errorlast + |T_{U_{\hat t_n}}|$ by replacing $\hatt_n$ with $t_n + \errorlast$, which leads to our second bound $\costalg \leq 2.5 \costopt + \errorlast$.
        \end{itemize}
\end{itemize}

Both cases show $\costalg \leq \min \{4 \costopt, 2.5 \costopt + |\errorlast|\}$, which completes the proof.
\end{proof}

\section{The Learning-Augmented Dial-a-Ride Problem}\label{LADARP}
Here we discuss how we extend the three models and algorithms to the OLDARP. 

\subsection{Problem Setting}
With a slight abuse of notation, we use the same notations as the TSP but change their definitions when the context is clear. The input of the OLDARP is a sequence denoted by $\X$ and with a size of $n$, the number of requests in a metric space. Each request in $\X$ is denoted by $x_i = (t_i, a_i, b_i)$ where $t_i$ is the time when the request becomes known, $a_i$ is the pickup position, and $b_i$ is the delivery position. A server starts and ends at the origin $o$, and it has to move each request $x_i$ from $a_i$ to $b_i$. Note that the server cannot collect a request $x_i$ before its arrival time $t_i$ ant let $c$ denote the maximum amount of requests the server can carry at a time. The goal is to minimize the completion time $|T_\X|$ of a route $T_\X$. In this work, we consider the server with unlimited capacity (i.e., $c = \infty$) and the non-preemptive version: once the server picks up a request $x_i$ at position $a_i$, it cannot drop it anywhere else except position $b_i$.
Regarding the previous results, Shmoys et al.~\cite{shmoys1995scheduling} proposed the non-polynomial \textsc{Ignore}~algorithm with
a competitive ratio of $2.5$ by ignoring newly arrived requests if the server is already following a schedule, while the \replan~algorithm achieved the same result by redesigning a route whenever a request becomes known. Ascheuer et al.~\cite{ascheuer2000online} improved the ratio to $2$, which meets the lower bound. 

\subsection{Models and Errors}
We consider the OLDARP using the three types of predictions discussed in the previous sections. First, for the model with sequence prediction without identity $\Xpred$, let $m$ denote the predicted size of the sequence and $\hat x_i = (\hat t_i, \hat a_i, \hat b_i)$ denote each predicted request for $i = 1, \ldots, m$. We do not quantify the error but only consider whether $\Xpred = \X$ or not. 
Next, we consider the sequence prediction with identity, $\Xpred$ with size of $n$. As in the TSP model, The time error is defined as $\errortime \defeq max_{i \in [n]}|\hat{t}_i - t_i|$. However, unlike the OLTSP, the OLDARP involves two positions. Thus, we modify the position error $\errorpos$ and define it as the extra distance the server has to travel to deliver the actual requests, compared with the predicted ones.
$$\errorpos \defeq \sum_{i} [ d(\hat a_i, a_i) + d(\hat b_i, b_i) ]$$
Last, we discuss the prediction of last arrival time $\hatt_n$, and the prediction error is defined by $\errorlast \defeq \hatt_n - t_n$, as in the TSP. 

\subsection{The \algdarT~Algorithm}
We introduce some properties of the \redesign~algorithm and a na\"{\i}ve algorithm, \algdarT, that we need before showing the details. First, recall that the \redesign~algorithm is $2.5$-competitive when ignoring the computational issue, while it is $3$-competitive in polynomial time. As Corollary~\ref{cor:PAHwait} for the \PAH~algorithm, we show that postponing the time to start the \redesign~algorithm does not affect the competitive ratio by Corollary~\ref{cor:replan}. 

\begin{theorem}[\cite{ascheuer2000online}, Theorem 3]\label{thm:replan}
The~\redesign~algorithm is $2.5$-competitive for the OLDARP in a metric space.
\end{theorem}

\begin{corollary}\label{cor:replan}
Assume the server operated by an algorithm waits at the origin $o$ for time $t$ before following the \redesign~algorithm and $t \leq t_n$. This algorithm is still $2.5$-competitive.
\end{corollary}

Then, we present 
the~\algdarT~algorithm for the DARP.
Since each request $\hat x_i = (\hat t_i, \hat a_i, \hat b_i)$ involves two positions, we describe a route as a sequence of positions $\hat T = (\hat p_{(1)}, \ldots, \hat p_{(2n)})$, where $\hat p_{(i)} \in \{\hat a_j\}_{j \in [n]} \cup \{\hat b_k\}_{k \in [n]}$ for $i \in [2n]$ and $\hat p_{(i)}$ denotes the $i^{th}$ position the server reaches in the route $\hat T$. Thus, the route $\hat T$ can show when the server pick up and deliver a request $\hat x_i$. In addition, the server adjusts the predicted route $\hat T$ when the actual requests in $\X$ arrive; when a request $x_i = (t_i, a_i, b_i)$ is released, we insert $a_i$ after $\hat a_i$ and $b_i$ after $\hat b_i$.

\begin{algorithm}
    \caption{\textsc{Learning-Augmented Dial-a-Ride Trust (LADAR-Trust)}}
    \label{alg:DARPnaive}
    \begin{algorithmic}[1] 
    \Require The current time $t$, the number of requests $n$, a sequence prediction $\Xpred$, and the set of current released unserved requests $U_t$.
    \State 
    First, compute an optimal route $\hat T = (\hat p_{(1)}, \ldots, \hat p_{(2n)})$ to serve the requests in $\Xpred$ and return to the origin $o$, where $\hat p_{(i)} \in \{\hat a_1, \ldots, \hat a_n, \hat b_1, \ldots, \hat b_n\}$ is the $i^{th}$ position the server reaches in $\hat T$.
    \State \textbf{For} any $i = 1, \ldots, n$ \textbf{do}
        \State \quad \textbf{If} $t = t_{(i)}$ for any $i$, where $x_{(i)} = (t_{(i)}, a_{(i)}, b_{(j)})$ \textbf{then}
            \State \quad \quad Update the route $\hat T$ by adding positions $a_{(i)}$ and $b_{(j)}$ after their corresponding 
            \Statex \quad \quad positions $\hat p_{(i)}$ and $\hat p_{(j)}$;
        \State \quad \textbf{EndIf}
        \State \quad \textbf{If} $p(t) = \hat p_{(i)}$ and $t < t_{(i)}$ \textbf{then}
            \State \quad \quad Wait at position $\hat p_{(i)}$ until time $t_{(i)}$.
            \Comment{Wait until the request arrives}
        \State \quad \textbf{EndIf}
    \State \textbf{EndFor}
    \end{algorithmic}
\end{algorithm}

We can extend Theorem~\ref{thm:naive} to the OLDARP by replacing the definition of $\errorpos = \sum^{n}_{i=1}d(\hat{p}_i, p_i)$ with $\errorpos = \sum_{i} [ d(\hat a_i, a_i) + d(\hat b_i, b_i)]$ and obtain the following result.

\begin{theorem}\label{lem:DARPnaive}
The competitive ratio of the \algdarT~algorithm is $1 + 2 \errortime + 4 \errorpos$. Thus, it is $1$-consistent and $(2 \errortime + 4 \errorpos)$-smooth but not robust.
\end{theorem}

\begin{proof}
As mentioned above, a route can be regarded as the order to reach certain positions. As in Theorem~\ref{thm:naive}, let $T_\X$ be the actual optimal route for the requests in $\X$ and $T_{\Xpred}$ be the predicted optimal route to visit the requests in $\Xpred$. We consider the total completion time of the following routes: 
(1) let $\costopt$ be the completion time of $T_\X$,
(2) let $Z^*_{\Xpred}$ be the completion time if the server visits the requests in $\Xpred$ by following the order of the route $T_\X$, 
(3) let $Z'_{\Xpred}$ be the completion time of $T_{\Xpred}$, and
(4) let $\costalg$ be the completion time of \algdarT, which visits the requests in $\X$ by following the order of the route $T_{\Xpred}$.
To relate $\costalg$ and $\costopt $, the proof is split into the following two parts: 

\begin{lemma}\label{clm:DAR1}
$Z^{'}_{\Xpred} \leq Z^*_{\Xpred} \leq \costopt + \errortime + 2 \errorpos$. 
\end{lemma}
\begin{proof}
To relate $Z^{'}_{\Xpred}$ and $\costopt$, we consider the following route $Z^*_{\Xpred}$ of server: 
(1) the server waits at the origin $o$ for time $\errortime$ before following the route $T_\X$, and (2) after the server arrives $a_i$ (resp. $b_i$), it moves to $\hat a_i$ (resp. $\hat b_i$) and goes back to $a_i$ (resp. $b_i$) to continue the route $T_\X$. 
We argue that the server can visit all of requests in $\Xpred$ and $\X$ then return to the origin at time $\costopt + \errortime + 2 \errorpos$. 
For (1), it guarantees that the server does not have to pay any extra waiting cost compared to the optimal route. For (2), it takes $2 d(a_i, \hat a_i) + 2 d(b_i, \hat b_i)$ to serve pick up or deliver request $\hat x_i$ then go back to the route $T_\X$. Combing (1) and (2), we can derive that the server can visit all of requests in $\Xpred$ and $\X$ then return to the origin at time $Z^*_{\Xpred} \leq \costopt + \errortime + 2 \errorpos$. Finally, since $T_{\Xpred}$ is the optimal route for $\Xpred$, we must have $Z^{'}_{\Xpred} \leq Z^*_{\Xpred} $.
\end{proof}

\begin{lemma}\label{clm:DAR2}
$\costalg \leq Z'_{\Xpred} + \errortime + 2 \errorpos$. 
\end{lemma}
\begin{proof}
To relate $\costalg$ and $Z'_{\Xpred}$, we consider the following route of server: (1) the server waits at the origin $o$ for time $\errortime$ before following the route $\hat T$, and (2) after the server arrives $\hat a_i$ (resp. $\hat b_i$), it moves to $a_i$ (resp. $b_i$) and goes back to $\hat a_i$ (resp. $\hat b_i$) to continue the route $\hat T$. We argue that the server can visit all of requests in $\Xpred$ and $\X$ then return to the origin at time $Z'_{\Xpred} + \errortime + 2 \errorpos$. For (1), it ensures that $x_i$ has arrived $p_i$ when the server is at $\hat p_i$. For (2), it takes $2 d(\hat a_i, a_i) + 2 d(\hat b_i, b_i)$ to serve request $x_i$ then go back to $\hat x_i$, which is exactly the operation of the \algdarT~algorithm. Therefore, we can derive that the server can visit all of requests in $\Xpred$ and $\X$ then return to the origin at time $Z^{ALG} \leq Z'_{\Xpred} + \errortime + 2 \errorpos$. 
The implication is that the server might need to adjust the route $\hat T$ if some request arrives later than expected, or if the position error exists.
\end{proof}

Finally, combining the results of Lemma~\ref{clm:DAR1} and Lemma~\ref{clm:DAR2}, we thus can obtain

\begin{align*}
    \costalg \leq \costopt + 2 \errortime + 4 \errorpos
\end{align*}
\end{proof}

\subsection{Sequence Prediction without Identity}
The server gets a sequence prediction $\Xpred$ with a size of $m$, the predicted number of requests, and an optimal route $\hat T$ to visit the requests in $\Xpred$ in the beginning. The idea of the \algdarNI~algorithm is to delay the predicted optimal route $\hat T$ to gain robustness. 
Specifically, the server first see whether the current moment $t$ is earlier than $\lambda |\hat T|$ or not. Before $\lambda |\hat T|$, the server follows the adjusted version of the \redesign~algorithm, which has a gadget: when the route $T_{U_t}$ is too long (i.e., $t + |T_{U_t}| > \lambda |\hat T|$), the server adjust the route $T_{U_t}$ to ensure that itself can be at the origin at time $\lambda |\hat T|$.
When and after $\lambda |\hat T|$, the server gets to follow the route $\hat T$ once there exists some unserved request (i.e., $U_t \neq \emptyset$). In the end, the server uses \PAH~to visit the remaining requests.
Note that if the server has to go back to the origin $o$ before unloading some requests, it carries all of them to the origin and deliver them in the following routes.

\begin{algorithm}
    \caption{\textsc{Learning-Augmented Dial-a-Ride Without Identity (LADAR-NID)}}\label{alg:DARPLASwo}
    \begin{algorithmic}[1] 
    \Require The current time $t$, a sequence prediction $\Xpred$, the confidence level $\lambda \in (0, 1]$, and the set of current released unserved requests $U_t$.
    \State 
    First, compute an optimal route $\hat T$ to serve the requests in $\Xpred$ and return to the origin $o$;
    \State \textbf{While} $t < \lambda |\hat T|$ \textbf{do}
        \State \quad \textbf{If} the server is at the origin $o$ (i.e., $p(t) = o$.) \textbf{then}
            \State \quad \quad Compute an optimal route $T_{U_t}$ to serve all the unserved requests in $U_t$ and return
            \Statex \quad \quad to the origin $o$;
            \State \quad \quad \textbf{If} $t + |T_{U_t}| >  \lambda |\hat T|$ \textbf{then}
            \Comment{Add a gadget}
                \State \quad \quad \quad Find 
                the moment $t_{back}$ such that $t_{back} + d(p(t_{back}), o) =  \lambda |\hat T|$;
                \State \quad \quad \quad Redesign a route $T'_{U_t}$ by asking the the server to go back to the origin $o$ at time
                \Statex \quad \quad \quad $t_{back}$ along the shortest path;
                \State \quad \quad \quad Start to follow the route $T'_{U_t}$.
            \State \quad \quad \textbf{Else}
                \State \quad \quad \quad Start to follow the route $T_{U_t}$.
            \State \quad \quad \textbf{EndIf}
        \State \quad \textbf{ElseIf} the server is currently 
        moving along a route $T_{U_{t'}}$, for some $t' < t$ \textbf{then}
            \State \quad \quad \textbf{If} a new request $x_i = (t_i, a_i, b_i)$ arrives \textbf{then}
            \Comment{Similar to \redesign}
                \State \quad \quad \quad Go back to the origin $o$.
            \State \quad \quad \textbf{EndIf}
        \State \quad \textbf{EndIf}
    \State \textbf{EndWhile}
    
    \State \textbf{While} $t \geq \lambda |\hat T|$ \textbf{then}
        \State \quad Wait until $U_t \neq \emptyset$;
        \State \quad Follow the route $\hat T$ until the server is back to the origin $o$;
        \State \quad Follow \redesign~($U_t$).
        \Comment{Serve the remaining requests}
    \State \textbf{EndWhile}
    \end{algorithmic}
\end{algorithm}

\begin{lemma}\label{cor:DARPconsistency}
The \algdarNI~algorithm is $(1.5 + \lambda)$-consistent, where $\lambda \in (0,1]$.
\end{lemma}
\begin{proof}
The server follows the adjusted version of the \redesign~algorithm at time $0$ and returns to the origin at time $\lambda |\hat T|$. At time $\lambda |\hat T|$, the server might follows the predicted route immediately if $U_{\lambda |\hat T|} \neq \emptyset$ or waits until the next request arrives. 
Assuming the prediction is perfect, i.e., $\Xpred = \X$ and $|\hat T| = \costopt$, the server completes serving requests after following $\hat T$.
Then, we consider the following cases:

\begin{itemize}
    \item {1. $U_{\lambda |\hat T|} \neq \emptyset$.} The server starts to follow route $\hat T$ at time $\lambda |\hat T|$, and the completion time is $\costalg = \lambda |\hat T| + |\hat T| \leq (1 + \lambda) \costopt$. 
    \item {2. $U_{\lambda |\hat T|} = \emptyset$.} After time $\lambda |\hat T|$, the server waits until the next request to come. Let $x_j = (t_j, p_j)$ denote this request. The completion time can be stated as $\costalg = \lambda |\hat T| + (t_j - \lambda |\hat T|) + |\hat T|$. Note that the online server waits at the origin $o$ for  $t_j - \lambda |\hat T|$. Then, we consider the duration $t_j - \lambda |\hat T|$ as the sum of the two time intervals according to the operations of the offline optimal server, i.e., OPT: $t_{move}$ denotes how much time the OPT is moving, and $t_{idle}$ denotes the length of time the OPT has nothing to do and can only wait. Since the server operated by the algorithm is also not visiting any request during $t_{idle}$, the case $t_{idle} \neq 0$ favors the algorithm. Thus, we can simply assume the worst case $t_{idle} = 0$. In addition, we know the distance the OPT saves during $t_{move}$ by moving earlier than the online server cannot be longer than $\frac{1}{2} \costopt$. Then, we have $t_{move} \leq \frac{1}{2} \costopt$, and the completion time is $\costalg \leq \lambda |\hat T| + t_{move} + |\hat T| \leq (1.5 + \lambda) \costopt$.
\end{itemize}
\end{proof}

\begin{lemma}\label{cor:DARProbustness}
The \algdarNI~algorithm is $(3.5 + 2.5 / \lambda)$-robust, where $\lambda \in (0,1]$.
\end{lemma}
\begin{proof}
We show the robustness is $3.5 + 2.5 / \lambda$ by the following cases.

\begin{itemize}
    \item {1. The server finishes before $\lambda |\hat T|$.} The server does what the original \redesign~algorithm would do. Thus, by Theorem~\ref{thm:replan}, the server obtains the bound $\costalg \leq 2.5 \costopt$.
    \item {2. The server cannot finish before $\lambda |\hat T|$.} After following $|\hat T|$, we cannot guarantee all requests in $\X$ have been served. Therefore, the server switches to the \redesign~algorithm to visit the remaining requests.
        \begin{itemize}
            \item{2-1. $U_{\lambda |\hat T|} \neq \emptyset$ and $t_n > (1 + \lambda)|\hat T|$.} We have $\costalg \leq 2.5 \costopt$ by Corollary~\ref{cor:replan}.
            \item{2-2. $U_{\lambda |\hat T|} \neq \emptyset$ and $t_n \leq (1 + \lambda)|\hat T|$.} The completion time is 
            \begin{align*}
                \costalg & \leq \lambda |\hat T| + |\hat T| + |T_{U_{(1+ \lambda)|\hat T|}}|
                \\& \leq 2.5 \costopt + (2.5 / \lambda) \costopt + \costopt
                \\& \leq (3.5 + 2.5 / \lambda) \costopt
            \end{align*}
            \item{2-3. $U_{\lambda |\hat T|} = \emptyset$  and $t_n > t_j + |\hat T|$.} Similarly, we have $\costalg \leq 2.5 \costopt$ by Corollary~\ref{cor:replan}.
            \item{2-4. $U_{\lambda |\hat T|} = \emptyset$  and $t_n \leq t_j + |\hat T|$.} The server completes visiting the requests at
            \begin{align*}
                \costalg & \leq t_j + |\hat T| + |T_{U_{t_j + |\hat T|}}|
                \\& \leq \costopt + (1 / \lambda) \costopt + \costopt
                \\& \leq (2 + 1 / \lambda) \costopt
            \end{align*}
        \end{itemize}
\end{itemize}
\end{proof}

\begin{theorem}\label{thm:DARPLASwo}
The~\algdarNI~algorithm is $(1.5 + \lambda)$-consistent and $(3.5 + 2.5 / \lambda)$-robust but not smooth, where $\lambda \in (0,1]$ is the confidence level.
\end{theorem}
\begin{proof}
We completes the proof by combining Lemma~\ref{cor:DARPconsistency} and Lemma~\ref{cor:DARProbustness}.
\end{proof}

\subsection{Sequence Prediction with Identity}
The server gets a sequence prediction $\Xpred$, which has the same size $n$ as the actual sequence $\X$. Since the arrival time of the last request provides a lower bound of the optimal completion time, the server can follow the route $\hat T$ earlier than in the previous model and still achieve robustness. Note that the route $\hat T$ can be described as a sequence of positions including $\hat a_i$ and $\hat b_i$ for any $i$. Before time $t_n$, the server follows the route $\hat T$ while making necessary adjustments: (1) if the server has reached the predicted position $\hat p_i$ but the request $x_i$ has not arrived, the server waits at $\hat p_i$ until $x_i$ comes, (2) when the request $x_i = (t_i, a_i, b_i)$ arrives, the server inserts position $a_i$ after $\hat a_i$ and position $b_i$ after $\hat b_i$ in the route $\hat T$. When the request $x_n$ arrives , the server knows it is the last one since the number of requests $n$ is given. Thus, the server weighs its two options: one is to continue on the route $\hat T$, whose distance is denoted as $r_1$, and the other is to go back to the origin and design a new route, denoted as $r_2$. Then, the server chooses the shorter one. 

\begin{algorithm}
    \caption{\textsc{Learning-Augmented Dial-a-Ride With Identity (LADAR-ID)}}\label{alg:DARPLASw}
    \begin{algorithmic}[1] 
    \Require The current time $t$, the number of requests $n$, a sequence prediction $\Xpred$, and the set of current released unserved requests $U_t$.
    \State $F = 0$;
    \Comment{Initialize $F=0$ to indicate that we trust the prediction at first}
    \State 
    First, compute an optimal route $\hat T = (\hat p_{(1)}, \ldots, \hat p_{(2n)})$ to serve the requests in $\Xpred$ and return to the origin $o$, where $\hat p_{(i)} \in \{\hat a_1, \ldots, \hat a_n, \hat b_1, \ldots, \hat b_n\}$ is the $i^{th}$ position the server reaches in $\hat T$.
    \State Start to follow the route $\hat T$;
    \State \textbf{While} $F = 0$ \textbf{do}
    \Comment{Trust the prediction}
        \State \quad \textbf{If} $t = t_{(i)}$ for any $i$, where $x_{(i)} = (t_{(i)}, a_{(i)}, b_{(j)})$ \textbf{then}
            \State \quad \quad Update the route $\hat T$ by adding positions $a_{(i)}$ and $b_{(j)}$ after their corresponding
            \Statex \quad \quad positions $\hat p_{(i)}$ and $\hat p_{(j)}$;
            \Comment{Update the route}
            \State \quad \quad \textbf{If} $t = t_n$ \textbf{then}
            \Comment{Find the shorter route}
                \State \quad \quad \quad $r_1 \leftarrow$ the remaining distance of following $\hat T$;
                \State \quad \quad \quad Compute a route $T_{U_{t_n}}$ to start and finish at the origin $o$ and serve the requests in
                \Statex \quad \quad \quad $U_{t_n}$;
                \State \quad \quad \quad $r_2 \leftarrow d(p(t), o) + |T_{U_{t_n}}|$;
                \State \quad \quad \quad \textbf{If} $r_1 > r_2$ \textbf{then}
                    \State \quad \quad \quad \quad Go back to the origin $o$;
                    \Comment{Give up the predicted route}
                    \State \quad \quad \quad \quad $F = 1$.
                \State \quad \quad \quad \textbf{Else}
                    \State \quad \quad \quad \quad 
                    Move ahead on the current route $\hat T$.
                \State \quad \quad \quad \textbf{EndIf}
            \State \quad \quad \textbf{Else}
                \State \quad \quad \quad 
                Move ahead on the current route $\hat T$.
            \State \quad \quad \textbf{EndIf}
        \State \quad \textbf{EndIf}
        
        \State \quad \textbf{If} $p(t) = \hat p_{(i)}$ and $t < t_{(i)}$ \textbf{then}
            \State \quad \quad Wait at position $\hat p_{(i)}$ until time $t_{(i)}$.
            \Comment{Wait until the request arrives}
        \State \quad \textbf{EndIf}
    \State \textbf{EndWhile}
    \State \textbf{While} $F = 1$ \textbf{do}
    \Comment{Do not rust the prediction}
        \State \quad Start to follow the route $T_{U_{t_n}}$ to serve the requests in $U_{t_n}$.
    \State \textbf{EndWhile}
    \end{algorithmic}
\end{algorithm}

\begin{theorem}\label{thm:DARPLASw}
The competitive ratio of the~\algdarI~algorithm
is $\min \{3, 1 + (2 \errortime + 4 \errorpos) / \costopt\}$. Thus, the algorithm is $1$-consistent, $3$-robust, and $(2 \errortime + 4 \errorpos)$-smooth.
\end{theorem}
\begin{proof}
The server follows an optimal route $\hat T$ before time $t_n$ while making adjustments. Then, at time $t_n$, it consider two available options. Note that $t_n \leq \costopt$ and $|T_{U_t}| \leq \costopt$ for any $t$.
\begin{itemize}
    \item{1. $r_1$.} The server continues on the route $\hat T$. This implies that the server follows the route $\hat T$ from beginning to end. By Theorem~\ref{lem:DARPnaive}, the completion time of the algorithm can be bounded by $\costalg \leq \costopt + 2 \errortime + 4 \errorpos$.
    \item{2. $r_2$.} The server goes back to the origin $o$ at time $t_n + d(p(t_n), o)$ and starts to follow a new route $T$, which visits all unserved requests in $U$. Since no more request arrives after $t_n$, the server completes at time $\costalg = t_n + d(p(t_n), o) + |T_{U_{t_n}}|$. In addition, the distance between the server and the origin is bounded by $d(p(t_n), o) \leq t_n$. Combined with the two inequalities mentioned above, we obtain the result $\costalg \leq 3 \costopt$.
\end{itemize}
As the server choose the shorter route among $r_1$ and $r_2$, the competitive ratio can be stated as $\min \{3, 1+(2 \errortime + 4 \errorpos) / \costopt\}$.

\end{proof}

\subsection{Prediction of the Last Arrival Time}
The idea is to be at the origin $o$ and starts the final route when the last request $x_n$ arrives. Thus, the server is return to the origin $o$ at the predicted last arrival time $\hatt_n$ and follow the \redesign~algorithm for the rest of the time. Whenever the server is at the origin $o$, it finds a route $T_{U_t}$ to visit the requests in $U_t$. If the route $T_{U_t}$ is too long, the server designs a shorter route $T'_{U_t}$ which would allow itself to be at the origin at time $\hatt_n$; otherwise, the server follows the route $T_{U_t}$ directly. If any request $x_i = (t_i, a_i, b_i)$ arrives, the server goes back to the origin $o$ and finds a new route.

\begin{algorithm}
    \caption{\textsc{Learning-Augmented Dial-a-Ride With Last Arrival Time (LADAR-Last)}}\label{alg:DARPLAM}
    \begin{algorithmic}[1] 
    \Require The current time $t$, the predicted last arrival time $\hatt_n$, and the set of current released unserved requests $U_t$.
    \State \textbf{While} $U_t \neq \emptyset$ \textbf{do}
    \State \quad \textbf{If} the server is at the origin $o$ (i.e., $p(t) = o$.) \textbf{then}
        \State \quad \quad Compute an optimal route $T_{U_t}$ to serve the requests in $U_t$ and return to the origin $o$;
        \State \quad \quad \textbf{If} $t < \hatt_n$ and $t + |T_{U_t}| > \hatt_n$ \textbf{then}
        \Comment{Add a gadget}
            \State \quad \quad \quad Find 
            the moment $t_{back}$ such that $t_{back} + d(p(t_{back}), o) = \hatt_n$;
            \State \quad \quad \quad Redesign a route $T'_{U_t}$ by asking the the server to go back to the origin $o$ at time
            \Statex \quad \quad \quad $t_{back}$ along the shortest path;
            \State \quad \quad \quad Start to follow the route $T'_{U_t}$.
        \State \quad \quad \textbf{Else}
            \State \quad \quad \quad Start to follow the route $T_{U_t}$.
        \State \quad \quad \textbf{EndIf}
    \State \quad \textbf{ElseIf} the server is currently 
    moving along a route $T_{U_{t'}}$, for some $t' < t$ \textbf{then}
        \State \quad \quad \textbf{If} a new request $x_i = (t_i, a_i, b_i)$ arrives \textbf{then}
        \Comment{Similar to \redesign}
            \State \quad \quad \quad Go back to the origin $o$.
        \State \quad \quad \textbf{EndIf}
    \State \quad \textbf{EndIf}
    \State \textbf{EndWhile}
    \end{algorithmic}
\end{algorithm}

\begin{theorem}\label{thm:DARPLAM}
The~\algdarL~algorithm
is a $\min \{ 3.5, 2 + |\errorlast| / \costopt \}$-competitive algorithm, where $\errorlast = \hatt_n - t_n$. Therefore, this algorithm
is $2$-consistent, $3.5$-robust, and $|\errorlast|$-smooth.
\end{theorem}

\begin{proof}
We consider the server's position when the last request $x_n$ arrives. We distinguish between two main cases depending on the predicted and actual last arrival time, $\hatt_n$ and $t_n$. 
Note that we get $t_n \leq \costopt$ trivially.

\begin{itemize}
    \item {1. $\hatt_n \leq t_n$.} The server is at the origin at time $\hatt_n$. After then, it receives some request and returns to the origin only when a request arrives or a route is completed. Therefore, the completion time is $\costalg = t_n + d(p(t_n), o) + |T_{U_{t_n}}|$, which leads to $\costalg \leq 2 \costopt + d(p(t_n), o)$. Note that we can extend Lemma~\ref{lem:1/2} to the~\redesign~algorithm
    on the DARP.
        \begin{itemize}
            \item {1-1.} By Lemma~\ref{lem:1/2}, we get $d(p(t_n), o) \leq 0.5 \costopt$ and our first bound $\costalg \leq 2.5 \costopt$. 
            \item {1-2.} The farthest position the server can be at time $t_n$ is $d(p(t_n), o) < t_n - \hatt_n = |\errorlast|$, since the server it at the origin at time $\hatt_n$. Our second bound for case 1 is $\costalg \leq 2 \costopt + |\errorlast|$.
        \end{itemize}
    \item {2. $\hatt_n > t_n$.} 
    Since the last request arrives before time $\hatt_n$ and the server could start its final route at $\hatt_n$ (if needed), we have $\costalg \leq \hatt_n + T_{U_{\hatt_n}}$.
    Let $t_L$ be the last time before $\hatt_n$ that a route is planned and $T^L$ be the original route planned at time $t_L$. Let $T'^L$ be the new route if $T^L$ is adjusted and equal to $T^L$ if not. Note that $t_L + |T'^L| \leq \hatt_n$ due to the choice of $t_{back}$.
        \begin{itemize}
            \item {2-1. $t_L + |T^L| < \hatt_n$.} The server finishes serving all requests in $\X$ before time $\hatt_n$ and thus $\costalg \leq \hatt_n$. In this case, since the algorithm runs in the same way as the \redesign~algorithm,
            we know that $\costalg \leq 2.5 \costopt$ by Theorem~\ref{thm:replan} still holds.
            \item {2-2. $t_L + |T^L| \geq \hatt_n$ and $t_n \leq t_L$.} The route $T^L$ is too long, and the algorithm find a new route $T'^L$ to let the server be at the origin $o$ at time $\hatt_n$. Since every request in $\X$ arrives before $t_L$, the server could have finished visiting the requests by following the original route $T^L$. However, the server chooses the adjusted route $T'^L$ instead. The additional cost of moving back to the origin is at most $2 d(p(t_{back}), o)$. Thus, by Lemma~\ref{lem:1/2} and Lemma~\ref{lem:replan}, the completion time is $\costalg \leq 2.5 \costopt + 2 d(p(t_{back}), o) \leq 3.5 \costopt$.
            \item {2-3. $t_L + |T^L| \geq \hatt_n$ and $t_n > t_L$.} The route $T^L$ is too long, and at least one request is unknown at time $t_L$.
            The server has to start a new route at $\hatt_n$, which is the last route since all requests have arrived by then.  Noting the route $T'^L$ also satisfies $|T'^L| \leq \costopt$, the completion time is 
            \begin{align*}
                \costalg &\leq \hatt_n + |T_{U_{\hatt_n}}|
                \\& = t_L + |T'^L| + |T_{U_{\hatt_n}}|
                \\& \leq t_n + |T'^L| + \costopt
                \\& \leq 3 \costopt
            \end{align*}
            Combining with the results in 2-1 and 2-2, we have our first bound $\costalg \leq 3.5  \costopt$.
            \item {2-4.} The server might have not reached the origin at time $t_n$. However, if the error $\errorlast$ is small, it must be close to the origin $o$ because of $t_{back}$. 
            Formally, we have $\costalg \leq t_n + \errorlast + |T_{U_{\hatt_n}}|$ by replacing $\hatt_n$ with $t_n + \errorlast$, which leads to our second bound $\costalg \leq 2 \costopt + \errorlast$.
        \end{itemize}
\end{itemize}

Both cases have shown $\costalg \leq \min \{3.5 \costopt, 2 \costopt + |\errorlast|\}$.

\end{proof}

\end{document}